\documentclass[11pt,a4paper]{article}
\usepackage{amsmath,amssymb,amsthm}
\usepackage{cite}
\usepackage{tikz}
\usepackage{bm}
\usepackage{graphicx} % To include figures

\usepackage{tikz}
\usepackage{placeins}
\usepackage{soul,xcolor}
\usepackage[utf8]{inputenc}
\usepackage{lipsum,textcomp}

\makeatletter
\theoremstyle{plain} 
\newtheorem{theorem}{Theorem}%[section] 

\theoremstyle{definition}               
\newtheorem{definition}[theorem]{Definition}
\newtheorem*{remark}{Remark}

\theoremstyle{plain}
\newtheorem{lemma}[theorem]{Lemma}
\newtheorem{Proposition}[theorem]{Proposition}
\newtheorem{Conjecture}[theorem]{Conjecture}
\makeatother

\definecolor{Mathematica1}{rgb}{0.368417, 0.506779, 0.709798}
\definecolor{Mathematica2}{rgb}{0.880722, 0.611041, 0.142051}
\definecolor{Mathematica3}{rgb}{0.560181, 0.691569, 0.194885}

\usepackage[pagebackref % this puts links to the page numbers where refs appear
]{hyperref}

\hypersetup{colorlinks=true, 
linkcolor=green!50!black,
citecolor=red!40!orange!90!black,
filecolor=OliveGreen, 
urlcolor=Mathematica2!80!orange!95!black,
filebordercolor={.8 .8 1}, 
urlbordercolor={.8 .8 0}}

\usepackage{multirow} % for multirow cells

\usepackage[shortlabels]{enumitem}
\usepackage{wrapfig}
\usepackage{setspace}
\usepackage[top=25truemm,bottom=30truemm,left=22truemm,right=22truemm]{geometry}

\definecolor{BSorange}{RGB}{140,50,0}

\newcommand{\calC}{{\mathcal C}}

\newcommand{\calI}{{\mathcal I}}
\newcommand{\calL}{{\mathcal L}}

\newcommand{\calS}{{\mathcal S}}

\newcommand{\calW}{{\mathcal W}}
 
\newcommand{\Tr}{{\rm Tr}}

\begin{document}

%\title{\bf Immersed figure-8 annuli and a strong isomorphism conjecture} 
%\title{\bfseries\Large Immersed figure-8 annuli and a strong isomorphism conjecture}
\title{\bfseries\Large Immersed figure-8 annuli and anyons}
\author{Bowen Shi\\[1mm]
\it\small Department of Physics, University of California at San Diego, La Jolla, CA 92093, USA}

\date{\today}

\maketitle

\begin{abstract}
    Immersion (i.e., local embedding) is relevant to the physics of topologically ordered phases through the entanglement bootstrap. An annulus can be immersed in a disk or a sphere as a ``figure-8", which cannot be smoothly deformed to an embedded annulus. We investigate a simple problem: is there an Abelian state on the immersed figure-8 annulus, locally indistinguishable from the ground state of the background physical system? We show that if the answer is affirmative, a strong sense of isomorphism must hold: two homeomorphic immersed regions must have isomorphic information convex sets, even if they cannot smoothly deform to each other on the background physical system. We explain why to care about strong isomorphism in physical systems with anyons and give proof in the context of Abelian anyon theory. We further discuss a connection between immersed annuli and anyon transportation in the presence of topological defects. In the appendices, we discuss related problems in broader contexts.
  \end{abstract}

\section{Introduction}

It is fascinating that many universal properties of a many-body system, such as topologically ordered phases, symmetry-protected phases, and quantum critical points, can be extracted from a single wave function \cite{kitaev2006topological,
levin2006detecting,
Calabrese:2009qy,
Li-Haldane2008,
Marvian2013,
Shapourian2017,
Zou2020,
Kim2021,
Fan2022Q,
Cian2022,
Lin2023}.
A further surprise is that in recent years, we have seen a hope to derive~\cite{Kim2013,Shi2018} the rules that govern those emergent theories from conditions on a state;
entanglement bootstrap \cite{shi2020fusion,Shi:2020domainwall,knots-paper,Shi2023-Kirby} is a framework that aims to do this. In entanglement bootstrap, we start with a single wave function (or reference state) on a topologically trivial region, e.g., a ball or a sphere. We impose some conditions on the wave function as the starting point and derive (bootstrap) laws of the emergent theory from there.  
On the way to deriving the emergent physical laws, we identify information-theoretic concepts (forms of many-body entanglement), such as information convex sets and the modular commutator. 
The goal of entanglement bootstrap has similarities and distinctions with quantum field theory, bootstrap in other physical contexts, renormalization group, categorical theory, and solvable models. Some of these aspects have been discussed in recent works on this subject.  

In this work, we investigate a simple aspect of entanglement bootstrap, the role of immersion. As we shall explain, immersion, i.e., local embedding of a topological space to another, is natural from the internal theoretical structure of the entanglement bootstrap. (Here, we are mainly interested in the case that a topological space is immersed in a background of the same space dimension. See the immersed ``figure-8" annulus in Fig.~\ref{fig:figure8-intro} for an example.) 
The importance of immersion is noticed only gradually: see \cite{Shi2023-Kirby} for a state-of-the-art discussion and \cite{Shi2019-Verlinde} for an early application with the terminology ``immersion" unnoticed. 
We suspect that the full extent of the benefits we can reap from immersion remains largely hidden beneath the surface. 

To motivate the usage of immersed regions, let us first recall why embedded regions are of interest in the study of topologically ordered systems with emergent anyonic excitations. 
On closed manifolds, such as tori, there can be multiple locally indistinguishable ground states. The information that distinguishes those ground states cannot be seen on any ball-shaped subsystem. Meaningful differences can, nevertheless, be detected if we examine subsystems with interesting topologies,   
such as annuli and (punctured) tori.    

In entanglement bootstrap, this intuition is made precise with the concept ``information convex set" $\Sigma(\Omega)$. It is a set of density matrices on $\Omega$, locally indistinguishable from the reference state. 
Originally, $\Omega$ is considered to be a subsystem, i.e., a region embedded in the physical system of interest. Smooth deformations of $\Omega$ must preserve the structure of information convex sets, according to the isomorphism theorem \cite{shi2020fusion}. Thus, only the ``topology class" of $\Omega$ is of interest.

Immersed regions are a broader class of regions. They are regions locally embedded in the physical system but need not be globally embedded. Because of this, states in immersed regions can use multiple copies of local Hilbert spaces of the original physical system. For instance, the figure-8 annulus can be immersed in a ball (Fig.~\ref{fig:figure8-intro}), and quantum states on the figure-8 annulus use the Hilbert space of the overlapping region (the region covered by two sheets) twice.  
Local embedding is sufficient for an information convex set to be defined \cite{knots-paper,Shi2023-Kirby,Shi2019-Verlinde} because we only needed local consistency of state in the information convex set with the reference state. Information convex sets are shown to be isomorphic under smooth deformation of the immersed region. These isomorphisms not only preserve the shape of the convex set but also preserve various distance measures between density matrices in the convex set.  

To our knowledge, the topological notion of immersion has rarely appeared in physics before. A notable exception is a sequence of connections observed by Hastings  \cite{Hastings:2013vma} and with Freedman~\cite{Freedman2019} on invertible phases and the classification of quantum cellular automata.

\begin{figure}[h]
    \centering
    \includegraphics[width=0.75\textwidth]{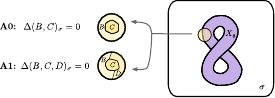}
    \caption{An immersed ``figure-8" annulus ($X_8$) on a background 2-dimensional physical system. The background physical system, often taken to be a ball or a sphere, is equipped with a reference state $\sigma$. For the reference state,  
    entanglement bootstrap axioms {\bf A0} and {\bf A1} are imposed on bounded radius balls such as the yellow ball illustrated. Entropy combinations $\Delta(B,C)$ and $\Delta(B,C,D)$ are defined in Eq.~\eqref{eq:Delta-def}.} 
    \label{fig:figure8-intro}
\end{figure}

What insight can we gain by considering immersed regions? 
One powerful application, in entanglement bootstrap of arbitrary dimensions, is the ability to construct very rich types of immersed punctured manifolds in an $n$-dimensional sphere~\cite{Shi2023-Kirby}. For instance, a punctured torus cannot be embedded in a 2-sphere, but it can be immersed in it. With the trick of healing the puncture, we obtain states on diverse types of closed manifolds, knowing only a state on a ball or a sphere! 
The ability to construct states on closed manifolds from immersion secretly agrees with that of a non-immersion approach mentioned by Kitaev in the context of invertible phases.\footnote{See \cite{Kitaev-2013} at around 1 hour and 35 minutes, where the class of manifolds is referred to as normally framed manifolds.}

Despite the progress mentioned above, the study of immersed regions in entanglement bootstrap is still in its early stages. Some topological manifolds can immerse in topologically distinct ways, even in a sphere. A simple example is the immersed ``figure-8" annulus, which cannot be smoothly deformed to an embedded annulus. One natural question is whether two homeomorphic regions $\Omega, \Omega'$ immersed in the $n$-dimensional sphere in inequivalent classes, still have isomorphic information convex sets (Conjecture~\ref{conj:2D-strong} and \ref{conj:strong}), which we call the strong isomorphism conjecture. One may also wonder if smooth deformation (and also ``tunneling" Sec.~\ref{sec:tunneling-trick}) between different immersed regions can give interesting automorphism of information convex sets. These two problems are closely related, as we shall explain, and both are open problems.

In the main body of this work, we shall focus on the simplest 2-dimensional (2D) context. The setup of the problem is recalled in Sec.~\ref{sec:background}. We note that figure-8 $X_8$ represents the only nontrivial class of immersed annulus on a sphere (Table~\ref{tab:figure-8}).  We explain that the strong isomorphism conjecture for 2D boils down to a simple question about the figure-8 annulus: ``is there an Abelian extreme point in $\Sigma(X_8)$?" See Sec.~\ref{sec:fun}, where the intuition is that an Abelian sector on figure-8 implies that tunneling a figure-8 (Sec.~\ref{sec:tunneling-trick}) not only changes the immersion class but also acts as an isomorphism between information convex sets. Such tunneling operations are powerful enough to relate homeomorphic immersed surfaces in different classes.

Is it true that the figure-8 annulus always detects an Abelian sector?
This problem appears to be nontrivial, and we provide a full solution only for Abelian anyon theories (Sec.~\ref{sec:Proof-Abelian}), i.e., when all the anyon types detected by the embedded annulus are Abelian. 
The more general case is left as a conjecture.

A crucial intuition we used in the proof is that anyons transported along the figure-8 annulus cannot be permuted. We discuss a thought experiment ``transportation experiment" (Sec.~\ref{sec:transportation}), that broadens the scope of this idea and applies it to physical contexts with topological defects. It relates the property of the immersed annulus that thickens the transportation loop to the question of whether anyons are permuted after the transportation. 
Many ideas and questions generalize to higher dimensions. In the appendix, we discuss a few useful tools for thinking about such generalizations.

\section{Background}\label{sec:background}

In this section, we review some necessary background for the entanglement bootstrap~\cite{shi2020fusion,Shi:2020domainwall,knots-paper,Shi2023-Kirby}. We review the setup and axioms, and then review a central concept, the information convex set. We then focus on the concept of immersed regions and related facts. 
In particular, we explain why there are precisely two classes of immersed annuli in a sphere; see Table~\ref{tab:figure-8}.  

\subsection{Setup and axioms}\label{subsec:axioms}

We consider the entanglement bootstrap setup in two spatial dimensions (2D). We consider a reference state $\sigma$ on a 2D surface $M$. Unless stated otherwise, we consider a reference state on the sphere $\mathbf{S}^2$, letting $M=\mathbf{S}^2$. (We used bold font to emphasize that a reference state is defined on the sphere.) In this case, $\sigma= |\psi\rangle \langle \psi|$.\footnote{Starting with a ball or a sphere is equivalent due to the ``completion trick" \cite{knots-paper}. We shall take $\mathbf{S}^2$ to be explicit. The important thing is that we do not need a state on a topologically nontrivial manifold, such as a torus.} 
We assume that the total Hilbert space is a tensor product of on-site Hilbert spaces from sites (or cells) of a coarse-grained lattice. The local Hilbert space for each coarse-grained site is finite-dimensional. 

The axioms {\bf A0} and {\bf A1} are imposed on small balls containing a few coarse-grained lattice sites. The regions are described in Fig.~\ref{fig:figure8-intro}, where
\begin{equation}\label{eq:Delta-def}
    \Delta(B,C,D) \equiv  S_{BC} + S_{CD} - S_B - S_D ,\quad \textrm{and} \quad  \Delta(B,C) \equiv \Delta(B,C,\emptyset).
\end{equation}
Here $S_A\vert_\rho$ means $S(\rho_A) = - \Tr(\rho_A \ln \rho_A)$ is the von Neumann entropy. The strong subadditivity \cite{Lieb1973} implies useful inequalities of similar (but diverse) forms. A particularly useful one to keep in mind is that, for any $\rho_{ABCD}$, $\Delta(B,C,D)_{\rho} \ge I(A:C|B)_{\rho}\ge 0$. Here $I(A:C|B) \equiv S_{AB} + S_{BC} - S_B - S_{ABC}$ is the conditional mutual information. Whenever the labels of regions are listed together, we mean the union of disjoint regions; for instance, $AB$ refers to the union of non-overlapping regions $A$ and $B$, where $A$ and $B$ do not share any qudits.

\begin{remark}
    The axioms {\bf A0} and {\bf A1} stated that the entropy conditions hold precisely. This can be a pretty strong requirement even though there are classes of solvable models that satisfy these conditions precisely \cite{Levin2005}. Nevertheless, there is evidence that precisely satisfied axioms are too strong to accommodate chiral topological orders, at least when the local Hilbert space on each site is finite-dimensional.\footnote{ One piece of evidence comes from \cite{Parent-Hamiltonian2024},  
    where a commuting projector parent Hamiltonian is identified. For a different argument that is based entirely on the quantum state, see~\cite{Li2024IMF}.}
    From a physical standpoint, the right way to appreciate these axioms is that we expect them to hold more accurately in larger systems. The first appearance of entropy in physics is in thermodynamics, where the prediction works best for large systems. We draw the parallel and argue that entanglement bootstrap predictions are also relevant for systems with approximate axioms. In fact, for chiral phases, a certain formula initially argued used precise area law~\cite{Kim2021} found confirmation with another theoretical approach~\cite{Fan2022} in the range of validity of the latter.
\end{remark}

The axioms are imposed on an intermediate length scale, which is much smaller than the total system size and larger than the correlation length and (microscopic) lattice spacing. This can be viewed as one instance of starting from the ``middle" \cite{Goldenfeld2023} rather than the ``bottom"~\cite{feynman1959plenty}, where entanglement bootstrap may be considered as one of the most quantum cases. Indeed, suppose the two axioms are satisfied on some intermediate scale. In that case, the same conditions hold on all larger length scales, and they define a renormalization group fixed point in this sense.  

\begin{remark}
    It is a conjecture that if the violations of the axioms are small enough, coarse-graining the lattice further will result in more accurate axioms. In other words, there may be a threshold of violation of the axioms, below which we expect the error of the axioms to decay with the length scale. See Chapter~11 of~\cite{ShiThesis} for a version of the conjecture and an explanation of why a threshold $0<\delta <\ln 2$ of violation of {\bf A1} is consistent with known examples from non-Abelian anyons, domain walls, defects, and spurious topological entanglement entropy.  
\end{remark}

\subsection{Information convex set} 

The information convex set is a central concept that is useful for understanding the main results of the paper. We will first recall its definition. Despite the effort given to the precise definition, we emphasize that knowing the precise definition is not required to understand the main result of the paper. An alternative route to read the rest of the paper is to treat the information convex set as a ``black box'', whose property is given by the isomorphism theorem and structure theorems we discuss. Readers taking this alternative route could still understand the main technical innovation without much loss. 

We say two density matrices $\rho_{AB}$ and $\lambda_{AB}$ are indistinguishable on subsystem $A$ if their reduced density matrices are identical on $A$, namely $\Tr_B \rho_{AB} =\Tr_B \lambda_{AB}$.
Alternatively, we say 
\begin{equation}
    \rho_{AB} \overset{c}{=\joinrel=} \lambda_A,
\end{equation}
which reads ``the states $\rho_{AB}$ is consistent with $\lambda_A$,'' where $\lambda_A=\Tr_B\lambda_{AB}$ is the reduced density matrix.  

The essence of information convex set ($\Sigma(\Omega)$) is a set of density matrices on region $\Omega$, such that any element is indistinguishable from the reference state $\sigma$ (introduced in section~\ref{subsec:axioms}) on a set of local regions.
The formal definition also involves the ability to thicken the region $\Omega$, and thus it is a bit subtle. In fact, there are several alternative definitions, but we choose not to discuss all of them.

We first state the definition of the information convex set for an embedded region $\Omega$ (Definition~\ref{def:ICS-embedded}). For the definition to be general, we denote the space manifold of the physical system as $M$. When the region $\Omega$ is embedded, we write  $\Omega \hookrightarrow M$. We say an embedded region is a subsystem. We then discuss the more general definition that is applicable to immersed regions (Definition~\ref{def:ICS-immersed}) following a formal definition of immersed region (Definition~\ref{def:immersed_region}).

\begin{definition}[Information convex set for embedded region \cite{shi2020fusion}]\label{def:ICS-embedded}
	For an embedded region $\Omega$, which can be thickened to $\Omega_+$, the information convex set $\Sigma(\Omega)$, for a given reference state $\sigma$, is the set of density matrices 
	\begin{equation}
		\Sigma(\Omega) \equiv \{ \rho_{\Omega} \vert\text{ conditions } 1,\,2 \},
	\end{equation}
where the two conditions are
\begin{enumerate}
	\setlength\itemsep{-0.2em}
	\item   $\rho_{\Omega}=\Tr_{\Omega_+ \setminus \Omega} \, \rho_{\Omega_+} $, where $\rho_{\Omega_{+}}$ is a density matrix on $\Omega_+$.
 
    \item  $\rho_{\Omega_+} \overset{c}{=\joinrel=} \sigma_b$ for any (bounded radius) ball $b \subset \Omega_+$. 
\end{enumerate} 
\end{definition}

\begin{figure}
    \centering
    \includegraphics[width=0.67\linewidth]{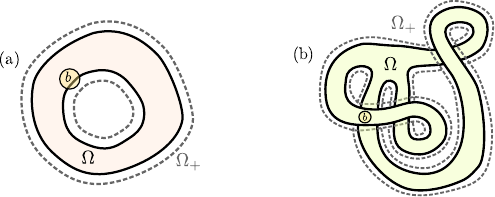}
    \caption{Regions used in the definition of information convex set. (a) $\Omega$ is an annulus embedded in the plane. $\Omega_+ \supset \Omega$ is its thickening. $b\subset \Omega_+$ is a bounded radius disk. (b) An immersed region $\Omega$, which is homeomorphic to a punctured torus.} 
    \label{fig:ICS-definition}
\end{figure}

The embedded regions $\Omega,\Omega_+$ and $b$ that appeared in the definition are illustrated in Fig.~\ref{fig:ICS-definition}(a). A few remarks are in order. 
\begin{enumerate}
    \item An embedded region can be thought of as a pair $\Omega=(\omega,\mathfrak{p})$ where $\omega$ is an abstract topological manifold of the same dimension as $M$ (possibly with boundaries).  $\mathfrak{p}:\omega \to M$ is the embedding map, which maps $\omega$ to $M$, such that the image is homeomorphic to $\omega$ through the map $\mathfrak{p}$. We can thus identify the embedded region with the image, and say $\Omega \subset M$. While this explanation may seem optional (as embedded regions are intuitively well-understood), we find it useful in the comparison with the formal definition of immersed regions below (Definition~\ref{def:immersed_region}). The information of the embedded region is \emph{less} than the embedding map because composing the embedding map with automorphism $\mathfrak{a}$ of $\omega$ (that is $\mathfrak{p}'=\mathfrak{p}\circ\mathfrak{a})$ defines the same embedded region. 
    \item The fact that the disk $b$ can be chosen from balls with a certain bounded radius is an important detail. It enabled the proof of various theorems in previous works~\cite{shi2020fusion,Shi:2020domainwall,knots-paper}. However, omitting the requirement of boundedness of radius still gives a valid alternative definition. 
   \item It follows from this definition that $\Sigma(\Omega)$ must be a compact convex set.
\end{enumerate}

\emph{Immersed regions:} Informally, immersed regions are regions that are locally embedded in the background manifold $M$.
We take the following definition of immersed regions. It is essentially the 2nd half of the definition of immersed regions in~\cite{knots-paper}, with a minor simplification.  

\begin{definition}[Immersed region]\label{def:immersed_region}
An immersed region is a pair
$\Omega=({\omega},\mathfrak{p})$. 
 Here $\omega$  
 be a topological manifold of the same dimension as the physical system $M$. $\mathfrak{p}: \omega\to  M  $ is an continuous map such that, for every point $x\in \omega$, there exists a disk-like neighborhood $n(x)$ such that when we restrict the domain of the map $\mathfrak{p}$ to $n(x)$, it gives an homeomorphism onto the image $b(\mathfrak{p}(x)) = \mathfrak{p} (n(x))\subset M$, where $b(\mathfrak{p}(x))$ is a bounded radius disk. 
  \end{definition}

 We use $\Omega \looparrowright M$ to indicate that $\Omega$ is a region that is immersed in $M$. When we say $\Omega$ is an annulus, we mean the topology of $\omega$ in the pair $\Omega=(\omega,\mathfrak{p})$ is an annulus. We shall refer to the continuous map $\mathfrak{p}$ obeying the requirement in Definition~\ref{def:immersed_region} as the \emph{immersion map}. For our purpose, replacing the immersion map by $\mathfrak{p}'=\mathfrak{p}\circ\mathfrak{a}$ give the same immersed region, where  $\mathfrak{a}$ is an automorphism of $\omega$.  
 
 Immersed regions can be specified by drawing a layered region ($\Omega$) onto the background manifold ($M$), such as the figure-8 annulus $X_8$ in Fig.~\ref{fig:figure8-intro}, and the punctured torus $\Omega$ in Fig.~\ref{fig:ICS-definition}(b). Note that such figures can determine the immersion map up to the automorphism (of $\omega$), which is precisely what we need. Embedded region is a special type of immersed region. 
  The advantage that immersed regions bring us is that we can sometimes use the qudits in the original physical system more than once. We recycle the qudit! 

  The precise meaning of recycling qubits is that we can define a pullback Hilbert space on the immersed region $\Omega=(\omega,\mathfrak{p})$. Suppose the original physical system has the total Hilbert space as a tensor product of onsite Hilbert spaces. We can find a set of sites on $\omega$ as the preimages of the physical sites under $\mathfrak{p}$. The Hilbert space on the immersed region $\Omega$ is defined to be the tensor product of the local Hilbert spaces on these sites. In general, the Hilbert space is larger than that associated with the image $\mathfrak{p}(\omega)\subset M$. The subsets of the image that are covered twice in Fig.~\ref{fig:ICS-definition}(b), each contributes two copies of the qudits. We also have the ability to pull back pieces of the reference state on a neighborhood $n(x)\subset \omega$. We shall denote such pull-back reference state as $\sigma_n^{[\mathfrak{p}]}$. Note, however, that we cannot directly pull back a ``global'' reference state on a nontrivially immersed $\Omega$; this is in contrast with what we can do for embedded regions.

 \emph{Information convex set for immersed regions:} Because the previous definition of an information convex set (Definition~\ref{def:ICS-embedded}) only requires comparing the density matrices on a set of bounded radius balls (neighborhood of points), it is reasonable to extend the definition to immersed regions. Indeed, this is possible.
 
\begin{definition}[Information convex set for immersed region~\cite{knots-paper}]\label{def:ICS-immersed}
	For an immersed region $\Omega \looparrowright M$, which can be thickened into immersed region $\Omega_+=(\omega_+, \mathfrak{p})$, the information convex set $\Sigma(\Omega)$, for a given reference state $\sigma$, is the set of density matrices 
	\begin{equation}
		\Sigma(\Omega) \equiv \{ \rho_{\Omega} \vert\text{ conditions } 1,\,2 \},
	\end{equation}
where the two conditions are
\begin{enumerate}
	\setlength\itemsep{-0.2em}
	\item   $\rho_{\Omega}=\Tr_{\Omega_+ \setminus \Omega} \, \rho_{\Omega_+} $, where $\rho_{\Omega_{+}}$ is a density matrix on $\Omega_+$.
 
    \item  $\rho_{\Omega_+} \overset{c}{=\joinrel=} \sigma^{[\mathfrak{p}]}_n$, for any disk-like neighborhood $n\subset \omega_+$ with pulled back reference state $\sigma^{[\mathfrak{p}]}_n$.  
\end{enumerate}  
\end{definition}
In practice, we may use the drawings of $\Omega \looparrowright M$ and $b\subset M$ [as in Fig.~\ref{fig:ICS-definition}(b)] to describe the second condition, i.e. the local consistency condition. Such a description omits the explicit mentioning of neighborhood $n$ and the immersion map $\mathfrak{p}$; this does require pointing to the layer in the drawing of $\Omega$ that $b$ talks to.

\subsection{Isomorphism, immersed annuli and its classification}\label{sec:immersed-annuli}

As explained in the previous section, immersed regions are regions locally embedded in the background physical system.    
For an immersed region $\Omega$, an information convex set $\Sigma(\Omega)$ can be defined (Definition~\ref{def:ICS-immersed}). 
Embedded regions, also called ``subsystems," are special cases of immersed regions. We shall denote ``$\Omega$ is immersed in $M$" by $\Omega \looparrowright M$. We are mainly interested in the case where the background physical system $M$ is the sphere.

We say two embedded regions $\Omega$ and $\Omega'$ are related by an \emph{elementary step} \cite{shi2020fusion} if we can partition the two regions as $\Omega=AB$ and $\Omega'=ABC$ in a relation shown in Fig.~\ref{fig:elementary-step}. Here $BCD$ is a partition of a bounded radius disk that we impose axiom {\bf A1}. Here, more precisely, we call $\Omega \to \Omega'$ an elementary step of extension and call $\Omega' \to \Omega$ an elementary step of restriction. The definition generalizes when the regions $\Omega$ and $\Omega'$ are immersed.

\begin{figure}[h]
    \centering
    \includegraphics[width=0.28\linewidth]{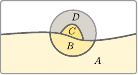}
    \caption{An illustration related to elementary steps. $A$ can be large and has an arbitrary topology, and it can be immersed. $BCD$ is a partition of a bounded radius (embedded) disk that we impose axiom {\bf A1}.}
    \label{fig:elementary-step}
\end{figure}

We say two immersed regions $\Omega$ and $\Omega'$ ($\Omega,\Omega'\looparrowright Y$) can be \emph{smoothly deformed} to each other on the background $Y$, when $\Omega$ can be deformed to $\Omega'$ by a finite sequence of elementary steps (of extensions and restrictions), each of these intermediate configurations is immersed in $Y$. We denote this situation as  $\Omega \overset{Y}{\sim} \Omega'$.
Note that $\overset{Y}{\sim}$ is an equivalence relation,\footnote{This equivalence is closely related to the regular homotopy between two immersion maps, considered in differential topology. If two immersed regions $\Omega=(\omega,\mathfrak{p})$ and $\Omega'=(\omega,\mathfrak{p}')$ are related by $\sim$, the two immersion maps $\mathfrak{p}$ and $\mathfrak{p}'$ can be chosen to be regular homotopic. However, we should modulo out automorphisms of $\omega$, which makes this different from a commonly studied object in math literature. However, this notion is similar to the notion of ``immersed surface" appeared in \cite{pinkall1985regular}.} because $\Omega \overset{Y}{\sim} \Omega$, and $\Omega \overset{Y}{\sim} \Omega'$ if and only if $\Omega' \overset{Y}{\sim} \Omega$.

There is an important theorem, called the \emph{isomorphism theorem}~\cite{shi2020fusion,knots-paper}, in entanglement bootstrap, under axioms {\bf A0} and {\bf A1}. The isomorphism theorem says for $\Omega,\Omega' \looparrowright \mathbf{S}^2$,  
\begin{equation}
    \Omega \overset{\mathbf{S}^2}{\sim} \Omega' \quad \Rightarrow \quad  \Sigma(\Omega) \cong \Sigma(\Omega').
\end{equation}
By isomorphism ``$\cong$", we mean the existence of an quantum channel $\Phi: \Sigma(\Omega) \to \Sigma(\Omega')$, which can be reversed by another channel $\widetilde{\Phi}: \Sigma(\Omega') \to \Sigma(\Omega)$. Such channels are completely determined by the sequence of elementary steps. As an isomorphism, some nontrivial structures of information convex sets are preserved, following from the existence of such channels:
\begin{itemize}
    \item [(i)] The geometry of the convex set is preserved. Namely, the extreme points of $\Sigma(\Omega)$ and $\Sigma(\Omega')$ are mapped to each other, and the map commutes with the convex combination:
    \begin{equation}
        \Phi(p \rho + (1-p)\lambda) = p\, \Phi(\rho) +(1-p) \Phi(\lambda), \quad\quad  \forall \rho,\lambda \in \Sigma(\Omega), \quad p\in [0,1].
    \end{equation}
    \item [(ii)] The distance measure between any two elements is preserved. For instance, the trace distance $D(\cdot, \cdot)$ and the fidelity $F(\cdot,\cdot)$ satisfy:
    \begin{equation}
        D(\rho,\lambda)= D(\Phi(\rho),\Phi(\lambda)),\quad F(\rho,\lambda)= F(\Phi(\rho),\Phi(\lambda)), \quad \forall \rho,\lambda \in \Sigma(\Omega).
    \end{equation}
    \item [(iii)]  The von Neumann entropy difference between any two elements is preserved:
    \begin{equation}
        S(\rho)-S(\lambda) = S(\Phi(\rho))- S(\Phi(\lambda)),\quad \forall \rho,\lambda \in \Sigma(\Omega).
    \end{equation}
\end{itemize}

\emph{Classification of immersed annuli:} As the isomorphism theorem suggests, what is interesting is the topological classification of immersed regions. Below, we shall be interested in immersed annuli.
We explain below that there are only two topological classes of immersed annuli, denoted by $X$ and $X_8$; see Table~\ref{tab:figure-8}.  (We adopt two alternative ways to draw $X_8$. One is as in Fig.~\ref{fig:figure8-intro}, where the overlap part is shown as two layers, one on top and covering another layer. The second way is as Table~\ref{tab:figure-8}, where the overlap part is transparent so that we see both layers. We note that the ``top and bottom" in the first way is introduced for drawing purposes only, and this is not part of the data to specify immersion. We shall frequently use the first approach as it simplifies drawing.) 
We shall be interested in knowing if $\Sigma(X) \cong \Sigma(X_8)$, which is not answered by the isomorphism theorem.

\begin{table}[h]
  \centering
  \begin{tabular}{|c|c|c|c|}
   \hline
    region's name & embedded annulus & figure-8 \\
    \hline
    symbol & $X$ & $X_8$  \\
    \hline
      \multirow{3}{*}{\centering region}
   &   \multirow{3}{*}{\raisebox{-0.08\height}
    {\includegraphics[width=1.45cm]{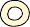}}} &   \multirow{3}{*}{\raisebox{-0.08\height}{\includegraphics[width=1.56cm]{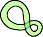}}  }
    \\
    &&\\
    &&\\
    \hline
    turning number & $ 1$ & $ 0 $  \\
    \hline
    superselection sectors & $\calC=\{1,a,b, \cdots \}$ &  $\calC_8=\{ \mu, \nu, \lambda, \cdots  \}$  \\
    \hline
  \end{tabular}
  \caption{The two topological classes of immersion of annuli in a sphere. One is represented by the embedded annulus $X$, and another is represented by the nontrivially immersed ``figure-8" $X_8$. Labels of superselection sectors, in the last line, follow from the simplex theorem reviewed in Section~\ref{sec:structure-theorem}.}\label{tab:figure-8}
\end{table}

According to Table~\ref{tab:figure-8}, figure-8 $X_8$ is the only nontrivial immersion type on a sphere. This fact follows from the following argument.
First, on the disk ($B^2$), each immersed annulus is the thickening of some curve in $B^2$, where the curve has no sharp corners. Therefore, they are classified by the turning number (modulo the sign). Two regions with different turning numbers cannot be converted to each other, as is well-known \cite{outside-in}. Some allowed and not allowed deformations on $B^2$ are illustrated as follows:
\begin{equation}
  \includegraphics[width =0.9 \textwidth]{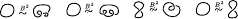}
\end{equation}
Here, each thick curve represents an immersed annulus, $\overset{B^2}{\sim}$ means two immersed regions can be smoothly deformed to each other on the background $B^2$, whereas ${\not\sim}$ means it is impossible to do such a deformation.  
On the sphere, there is an extra deformation that lets the annulus sweep over the entire sphere! This process can change the turning number by an even integer and thus further reduces the classification:
\begin{equation}
  \includegraphics[width =0.69 \textwidth]{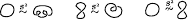}
\end{equation}
Consequently, we end up with two equivalence classes of immersed annuli. These two classes are represented by the embedded annulus $X$ and the immersed figure-8 annulus $X_8$, as described in Table~\ref{tab:figure-8}.

The central nontrivial claim of the isomorphism theorem is the \emph{existence} and uniqueness of a density matrix on another region that shares some property of an existing density matrix. This comes from a powerful ``merging technique", which enters the proof of the isomorphism theorem. The merging technique works even if we modify the topology of the region. In a later section, we shall review the merging technique when we use it in important ways.

\subsection{Superselection sectors and fusion spaces}\label{sec:structure-theorem}

Information convex sets detect superselection sectors and fusion spaces. In 2D, an embedded annulus $X$ detects the anyon types $\calC = \{1,a,b,\cdots\}$, and here $1$ denotes the vacuum sector.  The 2-hole disk (Fig.~\ref{fig:fusion-space} (a)) detects the fusion spaces. These follow from the structure theorems: the simplex theorem and the Hilbert space theorem in the original paper~\cite{shi2020fusion}.
The simplex theorem states that the information convex set of the embedded annulus, $\Sigma(X)$, is a simplex:
\begin{equation}
    \Sigma (X) = \left\{ \rho_X \left| \rho_X= \sum_{a\in \calC} p_a \, \rho^a_X \right.\right\},
\end{equation}
where the extreme points $\rho^a_X$ (which are density matrices\footnote{These extreme points are, in general, mixed states.}) are mutually orthogonal, that is, the fidelity between $\rho^a_X$ and $\rho^b_X$ is zero when $a \ne b$. $\{ p_a \}$ is a probability distribution, and the set of extreme point labels $\calC=\{1, a, b, \cdots\}$ is finite.

\begin{remark}
    More generally, the simplex theorem works for any sectorizable region~\cite{Shi:2020domainwall}. Sectorizable is a topological condition that says $S$ contains disjoint subsets $S_1$ and $S_2$, each of which can deform back to $S$ by a sequence of elementary steps of extensions. (The definition of ``elementary step'' is reviewed at the beginning of Section~\ref{sec:immersed-annuli}.) In 2D, immersed annuli and disks are the only sectorizable regions. A M\"obius strip embedded in $RP^2$ is not sectorizable, even though it is the thickening of some loop.
\end{remark}

Of importance is the fact that the simplex theorem works equally well for immersed regions, e.g., figure-8 $X_8$. The vacuum state is, however, defined \emph{only} for embedded annulus:
\begin{itemize}
    \item $\calC=\{1, a, b, \cdots \}$ is the superselection sectors defined from embedded annulus $X$. The vacuum ``1" is associated with a special extreme point, the reference state $\rho^1_X \equiv \sigma_X$.
    \item $\calC_8 =\{ \mu, \nu, \lambda, \cdots \}$ is the set of superselection sectors defined from the figure-8 $X_8$. We do not assume the existence of a vacuum.\footnote{The existence of the vacuum state for the figure-8 may be a logical consequence of the axioms. This remains an open problem.} 
\end{itemize}

Here are a few words about antiparticles. $\calC$ has a natural notion of antiparticle, that is, for each $a\in \calC$, there is a unique $\bar{a}\in \calC$ such that $a \times  \bar{a} = 1 + \cdots$; one way to define antiparticle is to deform the annulus $X$ back to itself in a nontrivial way on the sphere.  It is a further interesting question on whether $\calC_8$ has a natural definition of antiparticle, too. It seems natural to consider the following automorphism of $\Sigma(X_8)$:  
\begin{equation}\label{eq:figure-8-anti-attempt}
  \includegraphics[width =0.24 \textwidth]{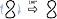}
\end{equation}
A more careful (and broader) discussion is postponed in Appendix~\ref{app:loop-space}.

A general immersed region ($Y$) can detect a fusion space $\mathbb{V}_I(Y)$.
Let $\partial Y$ be the thickened boundary\footnote{In entanglement bootstrap, a thickened (entanglement) boundary $\partial\Omega$ of a region $\Omega$, ($\partial\Omega \subset\Omega$), is the region $\Omega$ with its interior removed.} of $Y$. 
In general, $\partial Y$ can have multiple connected components, and each component is an annulus, either an embedded annulus (like $X$) or a nontrivially immersed annulus (like $X_8$). The label $I$ is then a collection of ordered labels from $\calC$ and $\calC_8$. Two instances are shown in Fig.~\ref{fig:fusion-space}, where $I=(a,b,c)$ for Fig.~\ref{fig:fusion-space}(a) and $I=(\mu,a, \nu)$ for Fig.~\ref{fig:fusion-space}(b). 
Once we specify the label $I$, we get a convex subset $\Sigma_I(Y)$ isomorphic to the state space (i.e., the space of density matrices) of a finite-dimensional fusion space $\mathbb{V}_I$:
\begin{equation}
    \Sigma_I(Y) \cong \calS(\mathbb{V}_I).
\end{equation}
This is the main content of the Hilbert space theorem. Furthermore, the thickened boundary $\partial Y$ is sectorizable, and it is useful to denote the set of superselection sectors as $\calC_{\partial Y}$ and note $I \in \calC_{\partial Y}$. 

\begin{figure}[h]
    \centering
    \includegraphics[width=0.4\textwidth]{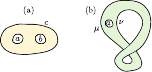}
    \caption{Examples of regions that detect fusion spaces. (a) A region that detects $\mathbb{V}_{ab}^c$, and (b) a region that detects $\mathbb{V}_{\mu a}^{\nu}$. Here, $a,b,c\in \calC$ and $\mu,\nu \in \calC_8$.}
    \label{fig:fusion-space}
\end{figure}

\subsection{Quantum dimensions}

We give a definition of quantum dimension for superselection sectors detected by immersed annuli.   
(A generalization into arbitrary spatial dimensions is presented in Appendix~\ref{app:equivalent-dim}.) Formally, the definition is as follows:

\begin{definition}[Quantum dimension] \label{def:d_h-2D}
Let $Z$ be an immersed annulus. Let $Z=BCD$, where $BD = \partial Z$ is the thickened boundary and $C=Z \setminus \partial Z$ is the interior; each of $B$ and $D$ is a union of two disks, as illustrated in Fig.~\ref{fig:quantum-dim-def}. For any extreme point $\rho^h_Z \in \Sigma(Z)$, 
\begin{equation}\label{eq:quantum-dim}
    d_h \equiv  \exp{\left( \frac{\Delta(B,C,D)_{\rho^h}}{4}\right) }
\end{equation}
is the quantum dimension of the superselection sector $h$.
\end{definition}

\begin{figure}
    \centering
 \includegraphics[width=0.5\textwidth]{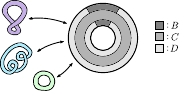}
    \caption{The partition of an immersed annulus $Z=BCD$ used in defining quantum dimension (Definition~\ref{def:d_h-2D}). We first map the immersed annulus to a topological annulus by a homeomorphism indicated by the arrow. We then partition the topological annulus as illustrated.}
    \label{fig:quantum-dim-def}
\end{figure}

Immediately following this definition, we have
\begin{enumerate}
    \item $d_h \ge 1$ for any $h \in \calC_Z$, where $\calC_Z$ is the set of superselection sectors characterized by $\Sigma(Z)$. This follows immediately from strong subadditivity $\Delta(B,C,D) \ge 0$. Note that $\calC_Z = \calC$ if $Z \overset{\mathbf{S}^2}{\sim} X$ and $\calC_Z = \calC_8$ if  $Z \overset{\mathbf{S}^2}{\sim} X_8$.
    \item Importantly, Definition~\ref{def:d_h-2D} does not require the knowledge (or existence) of a vacuum state on $Z$.
    On the other hand, for the embedded annulus $X$, where the vacuum is defined, Definition~\ref{def:d_h-2D} is equivalent to another definition $d_a \equiv \exp{(\frac{S(\rho^a_X) - S(\rho^1_X)}{2})}$.
    \item Definition~\ref{def:d_h-2D} works even if we relax the topology of the region that defines the reference state. This general setup allows  
    topological defects and is relevant to some discussion later in Section~\ref{sec:transportation} and in Appendix~\ref{app:loop-space}. 
\end{enumerate}

\begin{definition}[Abelian sector]
A superselection sector $h$ of an immersed annulus is Abelian if $d_h$ in Definition~\ref{def:d_h-2D} gives $d_h=1$.
\end{definition}

\begin{definition}[Abelian anyon theory]
    In the context of entanglement bootstrap in 2D, we say the theory is Abelian if $d_a =1, \forall a \in \calC$.
\end{definition} 

Note that the definition of Abelian anyon theory only puts a requirement on the embedded annulus.
For various reasons, we will need the following lemma.

\begin{lemma}[Abelian sector criterion]\label{lemma:d-TFRE}
 Let $\rho_Z \in \Sigma(Z)$, where $Z$ is an immersed annulus, the following are equivalent:
    \begin{itemize}
         \item [0.] $\rho_Z$ is an extreme point of $\Sigma(Z)$ associated with an Abelian sector.
        \item [1.] $\Delta(B,C,D)_{\rho} =0$ for the partition in Fig.~\ref{fig:quantum-dim-def}.
        \item [2.] $\Delta(B,C,D)_{\rho} =0$ for the partition in Fig.~\ref{fig:left-middle-right}(a).
        \item [3.] $\Delta(B,C,D)_{\rho} =0$ for the partition in Fig.~\ref{fig:left-middle-right}(b).
        \item [4.] $I(A:C|B)_{\rho} =0$ for the partition in Fig.~\ref{fig:left-middle-right}(c).
    \end{itemize}
\end{lemma}
According to Lemma~\ref{lemma:d-TFRE}, any one of the conditions 1, 2, 3, and 4 can be used to determine if a state in $\Sigma(Z)$ is an extreme point associated with an Abelian superselection sector.

\begin{figure}[h]
    \centering
    \includegraphics[width=0.67 \textwidth]{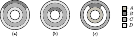}
    \caption{ As with Fig.~\ref{fig:quantum-dim-def}, we mapped the immersed annulus to the topological annulus by some (fixed) homeomorphism. Three distinct partitions of the annulus are shown in (a), (b), and (c).}\label{fig:left-middle-right}
\end{figure}

\begin{proof}
First, suppose that $\rho_Z$ is an extreme point of $\Sigma(Z)$. According to the definition of quantum dimension, $\Delta(B,C,D)_{\rho} =0$ for the partition in Fig.~\ref{fig:quantum-dim-def}.
It follows from Lemma~\ref{lemma:quantum-dim-TFRE} in Appendix~\ref{app:equivalent-dim} that any one of conditions 1, 2, 3, 4 implies the rest of them. (Lemma~\ref{lemma:quantum-dim-TFRE} in the appendix is a more general statement about the equivalence of definitions for quantum dimension, works beyond 2D setups.)  
Second, none of 1, 2, 3, 4 is true if $\rho_Z = \sum_h p_h \, \rho^h_Z$ is a non-extreme point.  
This is because
a non-extreme point has contribution $\sum_h p_h \ln \left(\frac{1}{p_h}\right)>0$ added to $\Delta$s  and $I$ in conditions 1, 2, 3, 4. This contribution is always positive for non-extreme points.  
From these two considerations, we derive Lemma~\ref{lemma:d-TFRE}.
\end{proof}

\section{Fun with figure-8}\label{sec:fun}

We have learned that the ``figure-8" region $X_8$ represents the only nontrivial class of immersed annulus in the sphere (Table~\ref{tab:figure-8}). It is arguably the simplest nontrivial immersed region and is fun to play with. The set of superselection sectors detected by $X_8$ are $\calC_8= \{ \mu,\nu,\cdots \}$. Throughout this section, we assume a reference state $\sigma = |\psi\rangle \langle \psi|$ on a sphere $\mathbf{S}^2$, and the axioms {\bf A0} and {\bf A1} are satisfied everywhere on the reference state. Also, recall that $X$ represents the embedded annulus, which can detect the anyon types $\calC=\{1, a, b, \cdots\}$.

The seemingly simple figure-8 region ($X_8$) already has interesting puzzles. The simplest one is on the existence of an Abelian sector in $\calC_8$: 

\begin{Conjecture}[Abelian sector on figure-8]\label{conj:1}
$\exists \mu \in \calC_8$, such that $d_\mu =1$. 
\end{Conjecture}
We emphasize that the conjecture is made in the context where the reference state is defined on the sphere $\mathbf{S}^2$. The axioms {\bf A0} and {\bf A1} are satisfied everywhere on the sphere. An alternative version of the conjecture requires the reference state $\sigma$ defined on an open disk. These two versions are equivalent due to the ability to ``complete" the disk into a sphere, as commented earlier.  
Note, however, that the conjectured statement is false if we remove the requirement of {\bf A1} centered at two or more points on the sphere. The counterexample is a state with a pair of topological defects.  If topological defects (e.g., that permute $e$ and $m$ in the toric code model, reviewed in Section~\ref{sec:defect-toric}) are present, then an annulus (either embedded or a figure-8) that winds the defect point around must have no Abelian sectors. 
We shall prove Conjecture~\ref{conj:1} for Abelian anyon theory in Section~\ref{sec:Proof-Abelian}.

Another natural question about figure-8 is whether or not
\begin{equation}\label{eq:iso-X-X8}
    \Sigma(X_8) \cong \Sigma(X).
\end{equation} 
The weaker statement $|\calC_8|=|\calC|$, i.e., the number of extreme points of $\Sigma(X_8)$ equals the number of anyon species, is true; we will discuss this in Section~\ref{sec:transportation}.  
If the two information convex sets $\Sigma(X_8)$ and $\Sigma(X)$ are isomorphic in the most natural manner, we expect an invertible map between them that preserves the quantum dimensions. Then, in addition to $|\calC| = |\calC_8|$, we should have an Abelian sector of $\calC_8$ that is mapped to the vacuum $1\in \calC$. This is one way to motivate Conjecture~\ref{conj:1}. (It is a further question whether there is a ``canonical vacuum" on $\calC_8$.)

An intriguing and nontrivial observation, which we shall make later in this section, is that Conjecture~\ref{conj:1} not only implies the seemingly stronger $\Sigma(X_8) \cong \Sigma(X)$ but also implies a very powerful ``strong isomorphism": 
\begin{Conjecture}[2D strong isomorphism conjecture]\label{conj:2D-strong}
\begin{equation}
    \Omega,\Omega' \looparrowright \mathbf{S}^2 \textrm{ are homeomorphic} \quad \Rightarrow \quad  \Sigma(\Omega) \cong \Sigma(\Omega').
\end{equation} 
\end{Conjecture} 

Importantly, we only required that $\Omega$ and $\Omega'$ are homeomorphic, and this is strictly weaker than $\Omega \overset{\mathbf{S}^2}{\sim} \Omega'$. (The weaker statement: $\Omega \overset{\mathbf{S}^2}{\sim} \Omega'$ implies $\Sigma(\Omega) \cong \Sigma(\Omega')$ is a known fact, from isomorphism theorem.) This motivates the name ``strong isomorphism conjecture".
As is evident, Eq.~\eqref{eq:iso-X-X8} is implied by the strong isomorphism conjecture. It is straightforward to generalize Conjecture~\ref{conj:2D-strong} to arbitrary spatial dimensions; see Appendix~\ref{sec:SIConj-arbitrary}. As a reminder, $\cong$ means there exists an isomorphism. It does not imply a ``canonical" isomorphism. $\cong$ preserves the entropy difference and fidelity between two elements.

\subsection{Tunneling trick}\label{sec:tunneling-trick}

We introduce a tunneling trick. The essential observation is that tunneling a figure-8 through a strip can add a twist to the strip. If the ``figure-8" region carries an Abelian sector, this process generates an isomorphism between information convex sets associated with homeomorphic topological regions in possibly distinct regular homotopy classes. 
We first explain the tunneling process at a topological level. After that, we discuss the details of entanglement bootstrap and explain what tunneling does to the quantum state (density matrices) and why an Abelian sector in $\calC_8$ plays a role.

 \begin{figure}[h]
     \centering
     \includegraphics[width=0.68\textwidth]{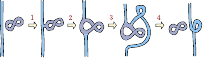}
     \caption{The topological detail of tunneling a figure-8 $X_8$ through a strip. 
     The strip can be thought of as a ``local piece" (in fact, a 1-handle) of some larger immersed region.}
     \label{fig:tunneling-main}
 \end{figure}

 At the topological level, the tunneling process contains a few steps. This is illustrated in Fig.~\ref{fig:tunneling-main}. In the first step, we attach $X_8$ to the strip by a small disk on the right. In the second and third steps, we deform the region such that $X_8$ disappears from the right and appears at the left of the strip, still attached. In the fourth step,  we detach $X_8$. 
 By the end of the tunneling process, a twist is added to the strip. 
It is instructive to compare the tunneling of $X_8$ with the tunneling of an embedded annulus $X$:
   \begin{equation}\label{eq:tunnel-trivial}
     \includegraphics[width=0.60\textwidth]{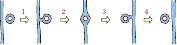}
 \end{equation}
 Tunneling an embedded annulus does not add any twist to the strip. The difference in the turning number between  $X_8$ and $X$ matters, and thus $X_8$ is of importance here.

 What does the tunneling process in Fig.~\ref{fig:tunneling-main} mean in the entanglement bootstrap? It represents a transformation of the quantum state (density matrix) on the region containing the strip. Let the region containing the original strip be $A$ and the region containing the strip after tunneling be $A'$, (see Fig.~\ref{fig:tunneling-EB}). Suppose, the state on ``figure-8" $C$ is an extreme point $\rho^{\mu}_{C}$, with $\mu \in \calC_8$. The transformation, as we shall discuss, must be a quantum channel that depends on $\mu$:
  \begin{equation}\label{eq:tunnel-channel}
     \Phi(\mu): \quad \Sigma(A) \to \Sigma(A').
 \end{equation}  
The quantum channel $\Phi(\mu)$ can be written as a product of three quantum channels 
\begin{equation} \label{eq:tunnel-channel-decompose}
    \Phi(\mu)= \Gamma^{\rm det}\circ\Phi^{\rm deform}_{ABC\to A'B'C'}\circ\Gamma^{\rm att}, 
\end{equation}
where the attaching map $\Gamma^{\rm att}: \Sigma(A) \to \Sigma_{[\rho_C^{\mu}]}(ABC)$ is a merging process, as we explain in the next paragraph. $\Sigma_{[\rho_C^{\mu}]}(ABC)$ is the subset of $\Sigma(ABC)$ with the constraint that the state reduces to $C$ gives $\rho^{\mu}_C$. The deformation map $\Phi^{\rm deform}_{ABC\to A'B'C'}: \Sigma_{[\rho_C^{\mu}]} (ABC) \to \Sigma_{[\rho_{C'}^{\mu}]}(A'B'C')$ is associated with the deformation (steps 2 and 3 of  Fig.~\ref{fig:tunneling-main}). This channel is an isomorphism (by the isomorphism theorem). The detaching map $\Gamma^{\rm det}: \Sigma_{[\rho_{C'}^{\mu}]}(A'B'C') \to \Sigma(A')$ is the partial trace $\Tr_{B'C'}$. In summary, the three quantum channels are 
\begin{equation}
    \begin{aligned}
        \Gamma^{\rm att}: &\quad \Sigma(A) \to \Sigma_{[\rho_{C}^{\mu}]}(ABC), \quad \qquad\qquad  \, \textrm{step 1 of Fig.~\ref{fig:tunneling-main}} \\
        \Phi^{\rm deform}_{ABC\to A'B'C'}: &\quad \Sigma_{[\rho_{C}^{\mu}]} (ABC) \to \Sigma_{[\rho_{C'}^{\mu}]}(A'B'C'), \qquad \textrm{step 2 and 3 of Fig.~\ref{fig:tunneling-main}} \\
        \Gamma^{\rm det}: & \quad  \Sigma_{[\rho_{C'}^{\mu}]}(A'B'C') \to \Sigma(A'), \qquad \qquad \, \textrm{step 4 of Fig.~\ref{fig:tunneling-main}}
    \end{aligned}
\end{equation}
In fact, any one of them is a sequence of Petz maps \cite{Petz2003} and partial traces.

  \begin{figure}[h]
     \centering
     \includegraphics[width=0.78\textwidth]{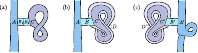}
     \caption{Analyzing the attachment (detachment) of an immersed figure-8 annulus $C$ ($C'$). Only part of $A$, $A'$ is shown. (a) A partition for merging. (b) and (c) have special usage when the state on ``figure-8" is an Abelian extreme point.} 
     \label{fig:tunneling-EB}
 \end{figure}

The attachment map $\Gamma^{\rm att}$ is defined as follows. Consider an arbitrary state $\lambda_A \in \Sigma(A)$ and an arbitrary extreme point $\rho_C^{\mu} \in \Sigma(C)$, with $\mu \in \calC_8$; the partitions are those in Fig.~\ref{fig:tunneling-EB}(a).
   First, we extend $\lambda_A$ to $\lambda_{AB} \in \Sigma(AB)$ and extend $\rho^\mu_C$ to $\rho^\mu_{BC} \in \Sigma(BC)$. Such an extension is unique because we enlarge the regions smoothly. $B=B_1B_2$.
   The two states obey the conditions for merging:
     \begin{equation}
         I(A:B_2|B_1)_{\lambda} =0, \quad  I(C:B_1|B_2)_{\rho^\mu} =0, \quad \Tr_A \,\lambda_{AB} = \Tr_C \, \rho^{\mu}_{BC}.
     \end{equation}
     According to the merging lemma \cite{Kato2016,shi2020fusion}, there is a unique state $\tau_{ABC}$ satisfying (i) $\tau_{ABC}$ reduces to $\lambda_{AB}$ on $AB$ and reduces to $\rho^\mu_{BC}$ on $BC$, (ii) $I(A:C|B)_\tau =0$. Then, by the merging theorem \cite{shi2020fusion,Shi:2020domainwall}, the state $\tau_{ABC}$ is in the information convex set $\Sigma(ABC)$, and in fact it is in the convex subset $\Sigma_{[\rho^{\mu}_C]}(ABC)$.
     The attachment map $\Gamma^{\rm att}$ is defined such that $\tau_{ABC} = \Gamma^{\rm att}(\lambda_A)$. Importantly, $\Gamma^{\rm att}$ is a sequence of Petz maps acting in a local neighborhood of the strip. It depends on $\mu$ through merging, and it is  {independent} of $\lambda_A$.

The discussion of the attachment map so far is for general $\mu \in \calC_8$. If $\mu$ is non-Abelian, the map $\Gamma^{\rm att}:\Sigma(A)\to \Sigma_{[\rho^{\mu}_{C}]}(ABC)$ may not be an isomorphism. 
 (The same is true already for tunneling an embedded annulus in the context of Eq.~\eqref{eq:tunnel-trivial}.)
The important thing, as we shall explain, is that 
 if $\mu \in \calC_8$ is Abelian, $\Gamma^{\rm att}$ and $\Gamma^{\rm det}$ are isomorphisms. 
 The special property of Abelian sectors relevant to $\Gamma^{\rm att}$ is that for any Abelian $\mu \in \calC_8$,
 \begin{equation} \label{eq:alpha}
     I(A:C|B)_{\alpha} =0, \quad \forall \alpha \in \Sigma_{[\rho^{\mu}_C]}(ABC).
 \end{equation}
 The derivation of Eq.~\eqref{eq:alpha} is as follows. We first extend the state $\alpha_{ABC} \in \Sigma_{[\rho^\mu_C]}(ABC)$ to a state in $\Sigma(ABCD)$, as in Fig.~\ref{fig:tunneling-EB}(b).
 If $\mu$ is Abelian, by Lemma~\ref{lemma:d-TFRE}, we have $\Delta(B,C,D)_{\rho^\mu}=0$.  
 Then by strong subadditivity, $I(A:C|B)_{\alpha} \le \Delta(B,C,D)_{\rho^\mu} =0$. 
 Because of its quantum Markov state structure, $\alpha_{ABC}$ can be recovered from $\alpha_{AB}$ by a quantum channel that is independent of $\alpha_{AB}$ (but depends on $\rho^\mu$). The fact that this $\alpha$-independent quantum channel reverses the partial trace $\Tr_C$ implies that $\Gamma^{\rm att}$ is an isomorphism. 
 For the same reason, the detachment map $\Gamma^{\rm det}$ is an isomorphism too. (The state exists by merging. The reversibility follows from $I(A':C'|B')=0$.)

In summary, tunneling a figure-8 region (Fig.~\ref{fig:tunneling-main}) induces a quantum channel of the form Eq.~\eqref{eq:tunnel-channel}. The channel $\Phi(\mu)$ is an isomorphism if $\mu \in \calC_8$ is Abelian.

 As an application, we see that $\Sigma(X_8) \cong \Sigma(X)$ if $\exists \mu \in \calC_8$ such that $d_\mu=1$, (that is if Conjecture~\ref{conj:1} holds). This is because we can tunnel a figure-8 through $X$ and turn $X$ into a figure-8 $X_8$. This result is summarized as follows:
 \begin{equation}\label{eq:tunnel-8-full}
     \includegraphics[width=0.4\textwidth]{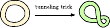}
 \end{equation}
Another way to illustrate this result uses the idea of handle decomposition; see e.g., Chapter~4 of~\cite{gompf1994}. An annulus is the union of a 0-handle with a 1-handle. The difference between the embedded annulus $X$ and the immersed ``figure-8" $X_8$ is the absence or existence of a twist on the 1-handle. We thus redraw \eqref{eq:tunnel-8-full} as: 
 \begin{equation}\label{eq:tunnel-8-handle}
     \includegraphics[width=0.46\textwidth]{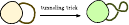}
 \end{equation}
 Here, a 0-handle is a disk (the round part of the figure), and a 1-handle is the strip (thick lines in the figure). Both 0-handles and 1-handles are topological balls. The difference is the way they attach.
 This illustration will be convenient.

\subsection{Abelian sector on ``figure-8" implies 2D strong isomorphism}

Now, we are in the position to explain why if the figure-8 annulus has an Abelian sector (i.e., if Conjecture~\ref{conj:1} is true), then the strong isomorphism conjecture (Conjecture~\ref{conj:2D-strong}) holds. This uses the tunneling trick discussed in Section~\ref{sec:tunneling-trick}. We explain why the trick is general enough.   

In the remainder of the section, we assume the existence of $\mu \in \calC_8$, such that $d_{\mu}=1$. The proof of strong isomorphism of an arbitrary region reduces to the study of connected regions pretty straightforwardly. Thus, we only consider connected regions $\Omega \looparrowright \mathbf{S}^2$ below. 
Because $\Omega$ is connected and is immersed in $\mathbf{S}^2$, it is either the sphere itself or it is an orientable surface with boundaries. Nothing needs to be proved for the case $\Omega$ is the sphere itself, and thus we consider the latter case. An orientable surface with boundaries can be written as a connected sum of $k\ge 0$ tori\footnote{The $k=0$ case, in our convention, is a sphere with $m$ holes.} ($\# k \,T^2$) with $m\ge 1$ holes; see, e.g. \cite{George-note-2011}. 
Thus, such an $\Omega$ is a \emph{boundary connected sum} of $k$ punctured tori $T^2\setminus B^2$ and $m-1$ annuli.

Components participating in the boundary connected sum (i.e., punctured tori and annuli) can be interchanged by sliding one across another; see the first deformation step\footnote{This is an analog of a well-known trick in topology, called handle slides, which is useful to arrange handles into a certain order (canonical form), see e.g., Fig.~1 of \cite{handle-slides-2015}.} of Fig.~\ref{fig:handle-slides}.  
Such a sliding trick works for both embedded regions and the more general immersed regions. After we arrange these components in the desired positions, we can deform them to 1-handles attached to the 0-handle, as shown in the second step in Fig.~\ref{fig:handle-slides}.

\begin{figure}[h]
    \centering
\includegraphics[width=0.79\linewidth]{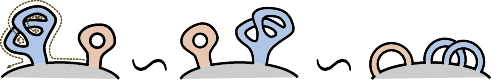}
    \caption{Step 1: We rearrange components in a boundary connected sum by sliding one component across another. The gray piece at the bottom is a large disk (0-handle). The way to slide the annulus (orange) from the right to the left across the punctured torus (blue) is indicated by the dashed line with an arrow.   Step 2: We deform each component to 1-handles attached to the 0-handle. This technique works when these components are immersed.}
    \label{fig:handle-slides}
\end{figure}

It follows that we can smoothly deform $\Omega$ on $\mathbf{S}^2$ such that it looks like the left figure of Fig.~\ref{fig:tunnel-general}. It is a 0-handle with $2k + m -1$ 1-handles attached in the order (around the 0-handle) indicated in the figure. Some of the 1-handles have twists. Upon the removal of possible twists, $\Omega$
 becomes $\Omega^\star$ in the right figure of Fig.~\ref{fig:tunnel-general}. \begin{figure}[h]
    \centering
    \includegraphics[width=0.86\textwidth]{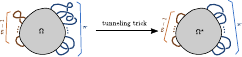}
    \caption{Tunneling trick applied to a connected region immersed in the sphere.}
    \label{fig:tunnel-general}
\end{figure}
In $\Omega^\star$, the 1-handles form two groups: $(m-1)$ of them are separated from others (dark brown), and there are $k$ pairs of 1-handles (dark blue), the two 1-handles in each pair ``cross" each other. Importantly, in both figures in Fig.~\ref{fig:tunnel-general}, we can arrange the 1-handles in the same order, circling around the 0-handle (the chunky gray disk), thanks to the ability to slide the components around each other, as explained in the previous paragraph.

Thus, for any connected $\Omega \looparrowright \mathbf{S}^2$, we can find a ``standard" immersed region $\Omega^\star \looparrowright \mathbf{S}^2$ in the homeomorphism class of $\Omega$ such that we can convert $\Omega$ to $\Omega^\star$ by the tunneling trick. Then, suppose there is an Abelian sector on ``figure-8", the tunneling trick implies the isomorphism
\begin{equation}
    \Sigma(\Omega) \cong \Sigma(\Omega^\star).
\end{equation}
Then, consider another region $\Omega' \looparrowright \mathbf{S}^2$ in the homeomorphism class of $\Omega$, we can do the same trick and show $ \Sigma(\Omega') \cong \Sigma(\Omega^\star)$. It follows that $\Sigma(\Omega) \cong \Sigma(\Omega')$. This completes the argument.

\section{Proofs in the context of 2D Abelian anyon theory}\label{sec:Proof-Abelian}

We can derive a concrete result (Proposition~\ref{Prop:Abelian-result}) for Abelian anyon theories in 2D. It implies that Conjecture~\ref{conj:1} and Conjecture~\ref{conj:2D-strong} are true for Abelian anyon theories.

\begin{Proposition}\label{Prop:Abelian-result}
    If $d_a =1, \forall a \in \calC$, then $d_\mu=1, \forall \mu \in \calC_8$. 
\end{Proposition}  

The proof uses Fig.~\ref{fig:Abelian-sketch}, where the key idea is that any anyon transported along the ``figure-8" does not change its type. Same with Section~\ref{sec:fun}, we would like to think of the figure-8 region in Fig.~\ref{fig:Abelian-sketch} to be embedded in either a ball or a sphere. However, it turns out that the argument only needs a certain property of the neighborhood of the figure-8 region. This is advantageous, and the same idea works in the transportation experiment in the next section.

\begin{figure}[h]
    \centering
    \includegraphics[width=0.97 \textwidth]{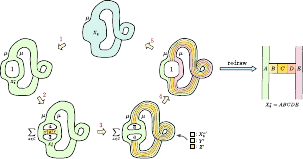}
    \caption{Operations on an extreme point $\rho^\mu_{X_8} \in \Sigma(X_8)$, where $X_8$ is an immersed figure-8 and $\mu\in \calC_8$. These are useful in proving Proposition~\ref{Prop:Abelian-result}, where we require $d_a=1, \forall a\in \calC$. The end result of the analysis is a useful condition $I(A:E|BCD)_{\rho^{\mu}}=0$ on the figure we redraw on the right; only part of the annuli $A$ and $E$ are shown.}
    \label{fig:Abelian-sketch}
\end{figure}

\begin{proof}
We consider a sequence of operations applying to a state on ``figure-8" $X_8$ carrying an arbitrary $\mu \in \calC_8$; see Fig.~\ref{fig:Abelian-sketch}.  We denote this extreme point of $\Sigma(X_8)$ as $\rho^\mu_{X_8}$. We draw $X_8$ in such a way that part of the ``figure-8" annulus is blown up. The region is like a ``watch with a twisted belt". We shall explain the steps in Fig.~\ref{fig:Abelian-sketch} and the fact that they imply a useful condition, Eq.~\eqref{eq:Abelian-CMI}, on the $\rho^{\mu}_{X_8}$. Below are the details.

In Step 1, we take a partial trace
    to create a hole in the interior of $X_8$. We call the resulting region $X_8^\circ$. The hole carries the vacuum sector $1\in \calC$. The two boundaries of the original figure-8 shape carry identical superselection sectors $\mu \in \calC_8$.
    In Step 2, we merge a disk ($YZ$) to the hole. The resulting region $X_8^\circ YZ$ is in a state
    \begin{equation}\label{eq:rho-mu}
        \tilde{\rho}^{\mu}_{X_8^\circ YZ} = \frac{1}{|\calC|} \sum_{a\in \calC} \rho^{a \bar{a} \mu \mu}_{X_8^\circ YZ}, \qquad I(X_8^\circ:Z|Y)_{\tilde{\rho}^\mu} =0,
    \end{equation}
    where the state $\tilde{\rho}^{\mu}_{X_8^\circ YZ}$ is an element in the information convex set $\Sigma(X_8^\circ YZ)$ and it is the indicated convex combination of extreme points $\rho^{a \bar{a} \mu \mu}_{X_8^\circ YZ}$ of $\Sigma(X_8^\circ YZ)$; each state $\rho^{a \bar{a} \mu \mu}_{X_8^\circ YZ}$ has the superselection sectors indicated on respective boundaries in the figure (after step 2). 

    In step 3, we smoothly deform the region ``figure-8-with-2-holes" back to itself. The deformation is such that the hole with anyon sector $a$ is transported along the figure-8 by almost a complete loop (initially moving upward), and the hole that carries $\bar{a}$ is also transported upwards slightly. Thus regions $Y$ and $Z$ are stretched during step 3 and become thin (and long) immersed strips $Y'$ and $Z'$.  
    After step 3, $\bar{a}$ is in the upper hole and $a$ is in the lower hole of ``figure-8-with-2-holes". 
    
    Although it might sound an obvious comment, we emphasize that transporting an anyon along the immersed figure-8 (embedded in a ball or a sphere) cannot change anyon type.\footnote{This is a special instance of Lemma~4.3 of \cite{shi2020fusion}, which says no matter how we transport an annulus within a fixed large disk (keep the annulus embedded during the process) and come back to its original position, anyons cannot be permuted.}
    It is because of this important property that we can conclude the argument. Suppose $X_8^\circ$, $Y$ and $Z$ are deformed into ${X_8^\circ}'$, $Y'$ and $Z'$ respectively after step 3. Then the state after step 3 is
    \begin{equation}\label{eq:rho-mu-pri}
        \tilde{\rho}'^{\mu}_{{X_8^\circ}' Y'Z'} = \frac{1}{|\calC|} \sum_{a\in \calC} \rho^{a \bar{a}\mu \mu}_{{X_8^\circ}' Y'Z'}, \qquad I(X':Z'|Y')_{\tilde{\rho}'^{\mu}} =0.
    \end{equation}
    In fact, $\rho^{a \bar{a}\mu \mu}_{{X_8^\circ}' Y'Z'} = \rho^{a \bar{a}\mu \mu}_{{X_8^\circ} YZ}$ because we deformed the "figure-8-with-2-holes" back to itself. (The uniqueness of the state with labeling $a\bar{a}\mu\mu$ follows from the fact that $a,\bar{a} \in \calC$ are Abelian.) Thus $\tilde{\rho}'^{\mu} = \tilde{\rho}^{\mu}$. The vanishing of conditional mutual information in \eqref{eq:rho-mu-pri} follows from that in \eqref{eq:rho-mu} and the fact that smooth deformation of the regions preserves the conditional mutual information.

    In Step 4, we trace out the horizontal strip, and this effectively brings $a$ and $\bar{a}$ together.  According to the Abelian fusion rule $a \times \bar{a} =1$ (derived in entanglement bootstrap in a self-contained way), the fusion outcome must be in the vacuum sector $1\in \calC$.
    The resulting state on $X_8^\circ = ABCDE$ is nothing but the reduced density matrix of $\rho^\mu_{X_8}$! (Step 5 is intended to be a dummy step indicating that it is possible to fill in the hole and get a state $\rho^\mu_{X_8}$. This is precisely the content of the previous sentence.) 
    Let us relabel the regions as the figure we redraw on the right, i.e., ${X_8^\circ}' \supset AE$, $Y'=BD$ and $Z'=C$.
    By the strong subadditivity  
    \begin{equation}
        I(AE:C|BD)_{\rho^{\mu}} \le I({X_8^\circ}':Z'|Y')_{\tilde{\rho}'^{\mu}}=0.
    \end{equation}
    We used the fact that $\tilde{\rho}'^{\mu}$, when reduced to $X_8^\circ$, gives $\rho^{\mu}_{X_8^\circ}$. % after we reduced it to $X_8^\circ$.
    Thus, $ I(AE:C|BD)_{\rho^{\mu}}=0$.

    Below, we show $I(AE:C|BD)_{\rho^{\mu}} =0$ implies a very useful condition
    \begin{equation}\label{eq:Abelian-CMI}
        \boxed{ I(A:E|BCD)_{\rho^{\mu}} =0.}
    \end{equation} 
    The steps are as follows.
    \begin{equation}
    \begin{aligned}  
      0 & = I(AE:C |BD)_{\rho^{\mu}} \\
        &= S_{ABDE} + S_{BCD} - S_{BD} - S_{ABCDE} \\
        & = (S_{AB} + S_{DE}) + (S_{BC} + S_{CD} - S_C) - (S_{B} + S_D) - S_{ABCDE} \\
        & = (S_{AB} + S_{BC} - S_B) + ( S_{DE} + S_{CD} - S_D) - S_C - S_{ABCDE} \\
        & = S_{ABC} + S_{CDE} - S_C - S_{ABCDE}\\
        & = I(AB:DE|C)_{\rho^{\mu}} \\
        & \ge I(A:E|BCD)_{\rho^{\mu}}.
    \end{aligned}
    \end{equation}
    In the third line, we use the fact that the extreme point $\rho^\mu_{X_8}$ factorizes on some region pairs such as $(AB,DE)$ and $(B,D)$ and the quantum Markov state condition $I(B:D|C)_{\rho^\mu}=0$. In the fifth line, we used $I(A:C|B)_{\rho^\mu} =0$ and $I(C:E|D)_{\rho^\mu}=0$. The last line follows from strong subadditivity.
    Thus,  Eq.~\eqref{eq:Abelian-CMI} holds.  
   Then, by Lemma~\ref{lemma:d-TFRE} (item 4 $\Rightarrow$ item 0), any $\mu \in \calC_8$ is Abelian.   
\end{proof} 

The key ingredient of the above proof is that transporting anyon along the figure-8 does not change its type; here, we assume that the figure-8 region is immersed in a ball or a sphere. Again, we emphasize that the crucial input is the anyon transportation property in the neighborhood of the immersed region, and it is not essential to assume the immersion is in a ball or sphere if we know this property. This observation will be useful later in Section~\ref{sec:transportation}.

The special property of Abelian anyon theory we make use of is the deterministic fusion outcome $a \times \bar{a} =1$. Because of this, the proof idea cannot be generalized easily to non-Abelian theories. One may wonder if a careful analysis of the fusion spaces involving multiple sectors in $\calC$ and $\calC_8$ can be helpful in the general case. A complementary idea is to try to take full advantage of the fact that the figure-8 annulus is immersed in a sphere. We leave the general case as an open problem.

\section{Transportation experiment}\label{sec:transportation} 

We discuss a thought experiment, ``transportation experiment," which reveals the physical connection between immersed annuli with topological defects. 
The entanglement bootstrap setup in this section is relaxed compared with that in Sections~\ref{sec:fun} and \ref{sec:Proof-Abelian}.  Let us consider 2D for simplicity. We allow the reference state to have ``defect points" where the axioms may be violated. (Alternatively, we consider a reference state defined on a manifold $N$ with an interesting topology, e.g., a torus or a $k$-hole disk.) The setup is illustrated in Fig.~\ref{fig:transport-intro}. 

\begin{figure}[h]
    \centering
    \includegraphics[width=0.88\textwidth]{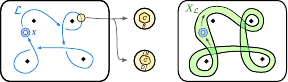}
    \caption{The entanglement bootstrap setup of Section~\ref{sec:transportation} and the main idea of the ``transportation experiment". The black dots in the figure are places we do not impose the axioms. The loop $\calL$ along which we transport the ``test annulus" (blue) is away from these black dots. The thickening of $\calL$ denoted by $X_{\calL}$ (green) is an immersed annulus.}
    \label{fig:transport-intro}
\end{figure}

We want to transport an anyon along a closed loop $\calL$, so that the anyon goes back to the original position in space. In entanglement bootstrap, this is done by transporting a ``test annulus" $X$, which is a sufficiently small embedded annulus that detects the anyon type by the information convex set $\Sigma(X)$.
Suppose we transport the test annulus $X$ along a closed loop $\calL$, so that it goes back to its original position, as illustrated in Fig.~\ref{fig:transport-intro}. We can assign an automorphism of the information convex set determined by the loop $\calL$:\footnote{Accurately speaking, the automorphism is determined by a loop $\calL$ with a chosen orientation. We may denote an oriented loop as $\vec\calL$. However, we shall use $\calL$ for simplicity. Quantities of interest below, such as the subset of anyons not permuted by the transportation, do not depend on the orientation. (For a discussion of automorphisms of information convex sets in broader setups, see Appendix~\ref{app:loop-space}.)}
\begin{equation}\label{eq:Psi-L}
    \Phi(\calL) : \Sigma(X) \to \Sigma(X).
\end{equation}
In general, the automorphism $\Phi(\calL)$ can permute extreme points of $\Sigma(X)$, and thus, it permutes anyon types. In other words, for each $\calL$, we have an automorphism of labels $\varphi_\calL: \calC \to \calC$. The automorphism  
is invariant under smooth deformations of loop $\calL$ that do not pass the defect points.   
Explicit examples come from topological defects, which we recall in Section~\ref{sec:defect-toric}. 
Throughout this section, we consider loops $\calL$ whose thickening ($X_\calL$) is an immersed annulus.\footnote{This excludes some loops on nonorientable manifolds. For instance, the noncontractible loop on $RP^2$ can be thickened to an M\"obius strip but not an annulus.}   
We define
\begin{equation}
    \calI(\calL) \equiv \left\{ a\in \calC | \, a = \varphi_\calL (a) \right\}
\end{equation}
as the subset of superselection sectors in $\calC$ that are preserved under the automorphism. Therefore, $|\calI(\calL)| = |\calC|$ if and only if no anyons are permuted.
We let $\calC_\calL$ be the set of superselection sectors associated with the extreme points of $\Sigma(X_\calL)$.

\begin{definition}[parallel transportation]
Transportation along $\calL$ is parallel if $|\calI(\calL)| = |\calC|$.
\end{definition} 
In this section, we ask the following questions.
\begin{equation}\label{eq:defect-question}
  \textrm{Question:} \,\, \qquad    |\calI(\calL)| = |\calC| \quad \overset{?}{\Longleftrightarrow}  \quad \exists \mu \in \calC_\calL  \,\, \textrm{ s.t. } d_\mu=1.
\end{equation} 
We explain why the ``$\Leftarrow$" direction is a relatively simple fact (Proposition~\ref{fact:parallel}). We conjecture that the statement is true in the ``$\Rightarrow$" direction (Conjecture~\ref{conj:transportation}). 
We explain the fact that ``$\Rightarrow$" is true for Abelian anyon theories, adopting the discussion in Section~\ref{sec:Proof-Abelian} to this context. Note that Conjecture~\ref{conj:transportation} implies the existence of an Abelian sector on the figure-8 annulus as a special case; thus, it also implies the strong isomorphism conjecture in 2D. 
Another result is Proposition~\ref{prop:parallel-not}, which states $|\calC_\calL| = | \calI(\calL) |$. It implies $|\calC_8| = |\calC|$ as a special case.

\subsection{Topological defects: an example}\label{sec:defect-toric}

A well-known example of a topological defect is the one that permutes $e$ and $m$ in the toric code model \cite{Bombin2010}. The anyon types of the toric code are $\calC=\{ 1, e, m, f\}$.  Their quantum dimensions are $d_1 = d_e = d_m = d_f =1$. 
Topological defects in this model are non-Abelian  $\calC'= \{\sigma_+,\sigma_-\}$ with $d_{\sigma_+} = d_{\sigma_-}= \sqrt{2}$.  Anyons $e$ and $m$ are permuted when they are transported around the defect (either $\sigma_+$ or $\sigma_-$). While defects are most familiar in the context where global symmetries are present~\cite{Barkeshli2014}, it is understood that the phenomenon of permuting anyons does not go away under arbitrary local perturbations, not preserving any symmetry.

\begin{figure}[h]
    \centering
    \includegraphics[width=0.52 \textwidth]{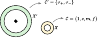}
    \caption{Information convex set is affected by the presence of a defect (black dot). The illustration is for the toric code with a defect that permutes $e$ and $m$.}
    \label{fig:defect-toric}
\end{figure}

The presence of a defect affects information convex sets.
This is illustrated in Fig.~\ref{fig:defect-toric} for the toric code example.
The annulus $X'$ surrounds the defect, and its information convex set $\Sigma(X')$ has two extreme points. The quantum dimension for each extreme point, defined according to Definition~\ref{def:d_h-2D}, agrees with $d_{\sigma_+}= d_{\sigma_-} = \sqrt{2}$. Thus, no Abelian sector exists on the annulus $X'$.
If we consider an embedded annulus $X$ not surrounding any defect, then the information convex set $\Sigma(X)$ is a simplex with four extreme points. These extreme points correspond to Abelian sectors $1$, $e$, $m$, and $f$. The model is solvable, and explicit computation\footnote{Similar computation can be done in some non-Abelian models, such as quantum double using the ``minimum diagram" technique \cite{knots-paper,Shi:2018krj}.} can be done to verify these statements. Putting aside the details, examples like this are nontrivial contexts for question \eqref{eq:defect-question}.

\subsection{Constraints of transportation}\label{sec:constraints-transport}

Below, we analyze the constraints of the transportation problem (Fig.~\ref{fig:transport-intro}). The entanglement bootstrap setup and notations have been described at the beginning of the section. 
As a useful fact, we note $1 \in \calI(\calL)$.  
This is because the test annulus $X$ remains embedded in the entire process, and we can reversibly fill in the hole at any step if it carries the vacuum sector $1\in \calC$. 
We start with a general statement:

 \begin{Proposition}\label{prop:parallel-not}
    $|\calC_\calL| = | \calI(\calL)|$. In particular, if the transportation is parallel, $|\calC_\calL| = |\calC|$.
\end{Proposition}
As a corollary, $|\calC|= |\calC_8|$. This is because transporting anyon along a figure-8 immersed in a ball or sphere cannot change the anyon type.  

Here is a sketch of the key idea of the proof of Prop.~\ref{prop:parallel-not}. The trick is to attach the test annulus to another annulus by a strip. After transporting the test annulus, the strip will be stretched along loop $\calL$. After gluing the test annulus with the other annulus (by merging with another annulus), we get a punctured torus $\calW_\calL$. This process is illustrated in Fig.~\ref{fig:transport-proof}. The main part of the proof is to analyze $\Sigma(\calW_\calL)$, or more conveniently $\Sigma(T^2)$, where the torus $T^2$ has a reference state obtained by an analog of ``vacuum block completion" (see Example~4.27 of \cite{Shi2023-Kirby}).  

\begin{proof} 
Consider the process illustrated in Fig.~\ref{fig:transport-proof}. We start with a disk away from the defect points (black dots) and remove two small balls from it. Let the resulting 2-hole disk be $Y$. Consider a state labeled by $a, 1, a$ on the three boundaries $\rho^{a1a}_Y \in \Sigma(Y)$. It must be an extreme point because the corresponding fusion space is one-dimensional.

In step 1, we move the annulus that carries $a\in \calC$ (the test annulus we choose) upwards so that the region $Y $ is deformed into an immersed region $Y'$. One of the boundaries of $Y'$ is now labeled by $\hat{1}$, indicating that it is the vacuum sector deformed onto an immersed annulus, so we cannot compare the state directly to the reference state. 
A fact useful later is that $\hat{1}$ is Abelian. In step 2, we transport the test annulus further along the closed loop $\calL$ and call the resulting region $Y_\calL$, as illustrated in Fig.~\ref{fig:transport-proof}(c).

 \begin{figure}[h]
    \centering
    \includegraphics[width=0.96 \textwidth]{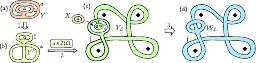}
    \caption{Illustrated is the construction of an extreme point in $\Sigma_{\hat{1}}(\calW_\calL)$, which is useful in the proof of Prop.~\ref{prop:parallel-not}.}
    \label{fig:transport-proof}
\end{figure}

If $a \in \calI(\calL)$, that is, $a$ is invariant under the transportation along $\calL$, we do step 3 of Fig.~\ref{fig:transport-proof}, which glues annulus $X$ with the state on $Y_\calL$ to a punctured torus $\calW_\calL$. In terms of quantum states, we merge two quantum states, which are extreme points of $\Sigma(X)$  and $\Sigma(Y_\calL)$ (thickening a certain region before merging if needed), with the superselection sectors labeled in Fig.~\ref{fig:transport-proof}(c). Here, $a \in \calI(\calL)$ is required so that the two states match on their overlap regions. For each valid choice of $a$, we get an extreme point of $\Sigma_{\hat{1}}(\calW_\calL)$.
Note that $1\in \calI(\calL) $, and therefore $\Sigma_{\hat{1}}(\calW_\calL)$ is nonempty.
Note that the punctured torus $\calW_\calL$ contains annuli $X_\calL$ and an embedded test annulus $X$.

Next, we ``complete" $\calW_\calL$ into a torus $T^2$ with a valid reference state. This step is optional, but it is both natural and convenient. The details are as follows. Choose $a=1$, and run the process in Fig.~\ref{fig:transport-proof} to obtain an extreme point $\rho^{|1_X\rangle}_{\calW_\calL} \in \Sigma_{\hat{1}}(\calW_\calL)$. By the completion trick (Lemma~4.4 of \cite{Shi2023-Kirby}), we first thicken $\calW_\calL$ and then heal the puncture. By doing this, we obtain a pure state $|1_X\rangle$ on a torus $T^2$. By construction, $T^2 \supset \calW_\calL$, and because $\hat{1}$ on the thickened boundary of $ \calW_{\calL}$ is Abelian, $|1_X\rangle$ is a valid reference state: it satisfies {\bf A0} and {\bf A1} everywhere. (The notation $|1_X\rangle$ is a reminder that if we reduce it to $X$, we get the reference state $\sigma_X$.)

Having explained the construction, we explain various constraints on the information convex sets, which leads to the conclusion.
First, by applying the associativity to step 3 (combining $X$ and $Y_\calL$ into $\calW_\calL$), we have $\dim \mathbb{V}_{\hat{1}} (\mathcal{W}_\calL) = |\calI(\calL)|$, this is by counting the number of ways to match the boundary conditions in Fig.~\ref{fig:transport-proof}. 
     If we take the reference state $|1_X\rangle$ on the ``completed" torus $T^2$, we further have
     $\dim \mathbb{V} (T^2) = |\calI(\calL)|$.
     This equality is easy to understand from the point of view that $T^2 = X \bar{X}$, where $\bar{X}$ is the complement of $X$ on $T^2$. Both $X$ and $\bar{X}$ are annuli and sectorizable. Applying the associativity to the gluing of the two annuli, we have 
    \begin{equation}\label{eq:XXbar}
       \dim \mathbb{V} (T^2) = \sum_{a,b \,\in \, \calC} \delta_{a,b}\, \delta_{a,\varphi_\calL(b)} =|\calI(\calL)|,
    \end{equation} 
    where $\delta_{a,b}$ comes from the matching of the sector on one shared boundary of the two annuli, and $\delta_{a,\varphi_\calL(b)}$ comes from the other shared boundary. The appearance of $\varphi_\calL$ is due to the transportation property.

 Similarly, we cut the torus along another direction into two annuli, $T^2 = X_\calL \bar{X}_\calL$. Here, $X_\calL$ is the immersed annulus that thickens $\calL$, a notation explained in Fig.~\ref{fig:transport-intro}. 
       Not only are $ X_\calL$ and $ \bar{X}_\calL$ both sectorizable, but they must also have isomorphic information convex sets $\Sigma(X_\calL) \cong \Sigma(\bar{X}_\calL)$.
       This is because $ X_\calL$ deforms to $ \bar{X}_\calL$ smoothly on $T^2$. So they both characterize $\calC_\calL$.   
      From a computation parallel to Eq.~\eqref{eq:XXbar}, we see
     \begin{equation}\label{eq:le-transport}
         \dim \mathbb{V}(T^2) \le |\calC_\calL|.
     \end{equation}
     The possibility to have ``$<$" is associated with the possibility that only a subset of labels in $\calC_\calL$ survive the matching and appear on the torus.     
     In particular, if ``$<$" applies, we cannot obtain the maximum-entropy state in $\Sigma(X_\calL)$ by a partial trace from any state in $\Sigma(T^2)$.

      Below, we show ``$=$" in Eq.~\eqref{eq:le-transport} must be realized. This is done by explicitly finding a state in $\Sigma(T^2)$ that reduces to the maximum-entropy state of $\Sigma(X_\calL)$. The state we identify is precisely $|1_X\rangle$. (The logic below is familiar. See, e.g., Lemma 5.7 of \cite{Shi2023-Kirby}.) 
     Recall that $|1_X\rangle$ reduces to the test annulus $X$ gives the vacuum state $\rho^1_X \equiv \sigma_X$, an Abelian extreme point.  
     Now apply the definition of quantum dimension \eqref{eq:quantum-dim} to the partition of $X=BCD$ as Fig.~\ref{fig:quantum-dim-def}, with the requirement that $B= \partial X \cap X_\calL$. By strong subadditivity,
     \begin{equation}
          0 = \Delta(B,C,D)_{|1_X\rangle }  \ge  I(\bar{X}:C|B)_{|1_X\rangle} \ge I(A:C_-|B)_{|1_X\rangle} ,
     \end{equation}
    for any $A \subset \bar{X}$ and $C_- \subset C$. The specific choice $A= X_\calL \setminus X$ and $C_- = X_\calL \cap (X \setminus \partial X)$ is what we want because
     $A B C_-$ is now a Levin-Wen partition \cite{Levin2005} of the immersed annulus $X_\calL$ (namely, $A$ and $C_-$ are two disks, and they are separated by $B$).   
       $I(A:C_-|B)_{|1_X\rangle}=0$ for the Levin-Wen partition \cite{levin2006detecting} of $X_\calL$ implies that $|1_X\rangle$ reduces to $X_\calL$ is the maximum-entropy state of $\Sigma(X_\calL)$. This implies ``$=$" applies to Eq.~\eqref{eq:le-transport}.  
\end{proof}

\begin{Proposition} \label{fact:parallel} 
If there exists $\mu \in \calC_\calL$ with $d_\mu =1$, the transportation along $\calL$ is parallel. 
\end{Proposition}
\begin{proof}
    The annulus $X_\calL$ can be viewed as a disk with a puncture. 
    If the state on $X_\calL$ is an extreme point, we can use the completion trick (Lemma~4.4 of \cite{Shi2023-Kirby}) to heal a puncture of $X_\calL$ to obtain a state on the disk. The resulting state satisfies the axioms on the entire disk if the extreme point is associated with an Abelian sector ($\mu \in \calC_\calL$ with $d_\mu=1$). Let this state be $\tilde{\sigma}$.
   It is known that transporting a test annulus on a disk cannot permute anyon types, supposing that the disk has a reference state. (This is by Lemma~4.3 of \cite{shi2020fusion}). As the loop $X_\calL$ is identified as part of the disk for which we have a valid reference state $\tilde{\sigma}$, the transportation of the test annulus on $X_\calL$ cannot permute anyons. Thus, the transportation along $\calL$ is parallel. 
\end{proof}

\begin{Conjecture}\label{conj:transportation}
If the transportation along $\calL$ is parallel, there exists $\mu \in \calC_\calL$ with $d_\mu =1$.
\end{Conjecture}

\begin{proof}[Proof of conjecture~\ref{conj:transportation} for Abelian anyon theories]
For Abelian anyon theory (that is $d_a =1$ for $\forall a\in \calC$), the conjecture holds. The proof is that of Proposition~\ref{Prop:Abelian-result}, with a few minor and easy-to-state modifications. We only need to replace the immersed figure-8 annulus in Fig.~\ref{fig:Abelian-sketch} with the immersed annulus $X_\calL$ in Fig.~\ref{fig:transport-intro}.
As we noted, every step works here because nowhere in the proof of Proposition~\ref{Prop:Abelian-result} did we use the condition that the annulus is immersed in a sphere. All we needed was the fact that a test annulus carrying $\forall a \in \calC$ transported along the immersed annulus remains in the same superselection sector $a$. For Abelian anyon theory, what we prove is $d_\mu =1$ for \emph{any} $\mu \in \calC_\calL$.
\end{proof}

We do not have proof of the general case. This conjecture implies Conjecture~\ref{conj:1}, and thus, it also implies the 2-dimensional strong isomorphism conjecture (Conjecture~\ref{conj:2D-strong}).

\section{Discussion}\label{sec:discussion}

In this work, we started by asking a simple question: Is there an Abelian superselection sector detected by the immersed figure-8  annulus?" The question is relevant to gapped topologically ordered systems in 2D, which support anyons and is formulated in the framework of entanglement bootstrap. We gave an affirmative answer to this question for Abelian anyon theories, but leave the general case as a conjecture (Conjecture~\ref{conj:1}). We explained the fact that if there is an Abelian sector on figure-8, a certain strong isomorphism  (Conjecture~\ref{conj:2D-strong}) of information convex sets must be true: if two homeomorphic regions $\Omega$ and $\Omega'$ are immersed in the sphere or a disk, the information convex sets $\Sigma(\Omega)$ and $\Sigma(\Omega')$ must be isomorphic. Importantly, it does not matter if $\Omega$ can be smoothly deformed to $\Omega'$ on the sphere (or disk), a situation where the established isomorphism theorem does not apply.

If the conjecture of the existence of an Abelian sector on ``figure-8" annulus
turns out false, a counterexample (e.g., a solvable model) would be extremely interesting. This also means the existence (absence) of Abelian sectors on immersed figure-8 will be a nontrivial (and exotic) criterion for classifying wave functions with entanglement area law as well as topologically ordered phases.

Suppose the existence of an Abelian sector on the figure-8 (Conjecture~\ref{conj:1}) can be proved. Strong isomorphism will be a fact. While strong isomorphism is powerful, it seems to imply that we cannot learn much by looking at a topological space immersed differently. Let me argue this is far from the case. If strong isomorphism holds, a couple of intriguing questions can be pursued. In particular, it means tunneling processes (Section~\ref{sec:tunneling-trick} and Fig.~\ref{fig:tunneling-main}) can generate explicit isomorphisms not allowed by smooth deformation. Thus, automorphisms of information convex sets will not only come from topologically nontrivial ways to smoothly deform an immersed region back to itself but also from ways that utilize tunneling in intermediate steps. We expect these diverse topological classes of automorphisms to be informative. A natural question is: ``Can we use these automorphisms (dancing of quantum states) to extract all the data of the modular tensor category underlying an anyon model?"\footnote{An optimism to this possibility was found independent by Kyle Kawagoe, who attended my chalk talk at math department Ohio State University.} If the answer is affirmative, it will manifestly use a single quantum state, and this will be complementary to the ``microscopic approach" \cite{Kawagoe2019}. Can the mathematical structure coming out of immersion be comparable to the mapping class group?  
Analogous questions in higher dimensions can be asked in parallel and are more open.

A stronger conjecture we postponed until now is the existence of a canonical (Abelian) vacuum state on the immersed figure-8 annulus.\footnote{In some classes of solvable models, such as the quantum double, indeed, such a unique Abelian vacuum on the figure-$8$ can be defined with a technique that depends on the special model. The challenge lies in the general definition. A toy version of the challenge already appears in the chiral semion chiral topological order (assuming the errors of the axioms do not affect the argument).} 
If this is true, the process in Fig.~\ref{fig:large-deformation} of the appendix may well be on the right track in extracting the full set of anyon topological spins. 
This is a meaningful problem that deserves an effort elsewhere.
Proving the existence of a canonical vacuum on figure-8 remains an open problem, even for Abelian anyon theories, in which context we proved Conjecture~\ref{conj:1}.

We further discussed the relevance of immersed annuli in contexts with topological defects. The central question we asked was Eq.~\eqref{eq:defect-question}, relating ``parallel" transportation to the existence of an Abelian sector on the immersed annulus that thickens the transportation loop. The question has two directions, and we provided an affirmative answer to one direction. The other direction is a conjecture that implies our main conjecture about the existence of an Abelian state on figure-8 (Conjecture~\ref{conj:1}) as well as strong isomorphism. 

Although this work focuses on 2D, we expect that the insights gained are useful in higher dimensions as well. Questions stirring in the right direction could be: what are generalizations of the tunneling trick (Section~\ref{sec:tunneling-trick} and Fig.~\ref{fig:tunneling-main}) and transportation experiments (Section~\ref{sec:transportation} and Fig.~\ref{fig:transport-intro}) in higher dimensions?  
 We developed some tools and discussed a few further open problems for higher dimensions in the appendices.

 Finally, we ask whether considering quantum states on regions immersed in a physical system can be useful in other physical contexts. Immersed regions utilize the Hilbert space more efficiently by using local pieces more than once, and they can be made larger and topologically (or geometrically) more interesting than the original physical system available to us. 
 Smooth deformation of immersed regions, a property natural in systems with entanglement area law, will be absent in broader contexts such as gapless systems. Nevertheless, many other benefits of immersion may persist.

\section*{Acknowledgments}
After settling the Abelian case, I made multiple attempts on the general proof of Conjecture~\ref{conj:1} without progress. These years gave me the time to contemplate the potential implications of this conjecture, especially strong isomorphism and tunneling. Conversations with many helped me sharpen these thoughts, including Anuj Apte, Yu-An Chen, Tarun Grover, Isaac Kim, Xiang Li, Ting-Chun Lin, Hanqing Liu, Daniel Ranard, Hao-Yu Sun, Chong Wang, Xueda Wen, and Carolyn Zhang.
I thank John McGreevy and Jin-Long Huang for collaborating on projects that taught me immersion-related topology. I thank Michael Levin for a conversation about topological spin during a UQM conference at Caltech in 2022 and an intriguing discussion with John McGreevy on a train to that conference. I thank my classmate Yang Fan for sharing the video ``outside in" \cite{outside-in} with me many years ago.
I thank Xinyao Zhao for providing a selection of colors, some of which I chose for some figures. I thank Ilya Gruzberg, Kyle Kawagoe, Yuan-Ming Lu, David Penneys, and Stuart Raby for related conversations during my visit to Ohio State University near the completion of this work. This work was supported by the Simons Collaboration on Ultra-Quantum Matter, a grant from the Simons Foundation (652264).

\appendix

\section{Equivalent definitions of quantum dimension}\label{app:equivalent-dim}

Consider the context of $n$-dimensional entanglement bootstrap.
We derive an equivalence relation on the definition of quantum dimensions (Lemma~\ref{lemma:quantum-dim-TFRE}) for immersed sphere shells in $n$-dimensions, where $n \ge 2$ and is arbitrary.  
It applies to immersed annuli as a special case ($n=2$).

Because the setup is general, we introduce some diagrammatic notation for the partitions.  
Let $Z = S^{n-1}\times \mathbb{I}$ be a sphere shell in $n$ space dimensions, where $S^{n-1}$ is a $(n-1)$-sphere and $\mathbb{I}$ is an interval. In Fig.~\ref{fig:TFRE}, the horizontal direction indicates the direction of the interval $\mathbb{I}$. Therefore, each vertical strip that connects the top and the bottom represents a sphere shell. If a strip is cut in half (top and bottom), each piece is a half-sphere shell (topologically a ball).
Note that $Z$ is sectorizable, and thus $\Sigma(Z)$ is a simplex. We denote the superselection sectors associated with the extreme points of $\Sigma(Z)$ as $\calC_Z$.

\begin{figure}[h]
    \centering
    \includegraphics[width=0.86\textwidth]{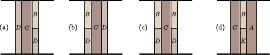}
    \caption{Partitions of a sphere shell $Z= S^{n-1} \times \mathbb{I}$. The horizontal direction is the direction of the interval $\mathbb{I}$.} 
    \label{fig:TFRE}
\end{figure}

We start with a useful fact. An element $\rho_Z$ of $\Sigma(Z)$ is an extreme point if and only if %, where $h \in \calC_Z$:
\begin{equation}
      \Delta(\partial Z, Z\setminus \partial Z)_{\rho}=0.
\end{equation}
Here $\partial Z$ is the thickened boundary of $Z$, which is $BD$ in the partition Fig.~\ref{fig:TFRE}(a). $Z\setminus \partial Z$ is the interior of $Z$, which is $C$ in the same partition. This condition is the
 \emph{extreme point criterion}~\cite{knots-paper,Shi:2020domainwall}.

\begin{lemma}\label{lemma:quantum-dim-TFRE}
    The following are equivalent definitions of quantum dimension $d_h$ of $h\in \calC_Z$, where $Z$ is an immersed sphere shell in $n$-dimensional space.
\begin{enumerate}
    \item [(a)] $\Delta(B,C,D) = 2 \ln d_h$ for the partition in Fig.~\ref{fig:TFRE}(a).
    \item [(b)] $\Delta(B,C,D) = 2 \ln d_h$ for the partition in Fig.~\ref{fig:TFRE}(b).
    \item [(c)] $\Delta(B,C,D) = 4 \ln d_h$ for the partition in Fig.~\ref{fig:TFRE}(c).
    \item [(d)] $I(A:C|B) = 2 \ln d_h$ for the partition in Fig.~\ref{fig:TFRE}(d).
\end{enumerate}
\end{lemma}

\begin{remark}
    We shall see in the proof that we do not need to deform the annulus. Therefore, Lemma~\ref{lemma:quantum-dim-TFRE} holds not only for sphere shells immersed in a sphere but also is applicable to the context where the sphere shells are immersed in other topological manifolds.  
\end{remark}

\begin{proof}
    First, we consider the partition in Fig.~\ref{fig:TFRE-proof}(a), where $Z=A B C D E F$.
    By the strong subadditivity, $\Delta(B,C,DE)_{\rho^h} \ge I(A:C|B)_{\rho^h}$. Furthermore,
 because $\rho_Z^h$ is an extreme point of $\Sigma(Z)$, it factorizes as  $\rho_{B C D E F}^h = \rho_{B C D E}^h\otimes \rho_F^h$, which implies
 \begin{equation}\label{eq:use1}
      \Delta(B, C, D E)_{\rho^h}=\Delta(B, C, D E F)_{\rho^h}.
   \end{equation} 
The extreme point $\rho^h_Z$ of $\Sigma(Z)$, satisfies $\Delta(\partial Z, Z\setminus \partial Z)_{\rho^h}=0$, by the ``extreme point criterion". It
implies
\begin{equation}\label{eq:use2}
\begin{aligned}
    S_{C D E F} &= S_Z+S_{A B},  \\
  S_{D E F} & =S_Z+S_{A B C},
\end{aligned}
\end{equation}
where the first line follows letting $\partial Z$ be $CDEF$. The second line uses $0 = \Delta(EF, ABCD) \ge \Delta(DEF, ABC)$.

    \begin{figure}[h]
    \centering
    \includegraphics[width=0.82\textwidth]{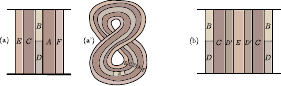} 
    \caption{Partitions used in the proof of Lemma~\ref{lemma:quantum-dim-TFRE}. In (a) and (b), the convention is the same as that in Fig.~\ref{fig:TFRE}. (a') illustrates the partition of (a) in the concrete case that $Z=X_8$ is the ``figure-8".} 
    \label{fig:TFRE-proof}
\end{figure}

 Therefore, combining \eqref{eq:use1} and \eqref{eq:use2}, we get
 \begin{equation}
     \begin{aligned} 
     \Delta(B, C, D E)_{\rho^h} =&\Delta(B, C, D E F)_{\rho^h} \\ 
     =& \left(S_{B C}-S_B+S_{C D E F}-S_{D E F}\right)_{\rho^h} \\ 
     =& \left(S_{B C}-S_B+S_{A B}-S_{A B C}\right)_{\rho^h}\\
     =& I(A : C | B)_{\rho^h}.
     \end{aligned}
 \end{equation}
Thus, among the statements in Lemma~\ref{lemma:quantum-dim-TFRE}, (a) and (d) are equivalent. Similarly, (b) and (d) are equivalent. Therefore, the three statements (a), (b), and (d) in Lemma~\ref{lemma:quantum-dim-TFRE} are equivalent. 
 
The next step is establishing the equivalence between statement (c) and the rest. For this purpose, we consider Fig.~\ref{fig:TFRE-proof}(b), where $Z=BCDD'E$. We will show:
\begin{equation}\label{eq:use3}
    \Delta(B,CD'E,D)_{\rho^h} =   \Delta(B,C,DD')_{\rho^h}.
\end{equation}
Then, because the extreme point $\rho^h_Z$ factorizes as a tensor product on the left of $E$ and right of $E$, we have $\Delta(B,CD'E,D)_{\rho^h} = \Delta_L + \Delta_R$. $\Delta_L$ and $\Delta_R$ give contributions identical to conditions (a) and (b) of Lemma~\ref{lemma:quantum-dim-TFRE}.  Thus, the claimed result holds.  

The rest of the proof is the justification of \eqref{eq:use3}. In fact, this is a special case of the decoupling lemma (Lemma D.1 of \cite{knots-paper}). Nonetheless, we spell out the details so readers do not need to read the more general decoupling lemma.

All regions below are those in Fig.~\ref{fig:TFRE-proof}(b).
Because $\rho^h_Z$ is an extreme point, it satisfies
\begin{equation}\label{eq:useful-for-da}
    \begin{aligned}
        S_{BCD' E} &= S_{BC}  - S_C + S_{CD' E} \\
        S_{CDD'E} & = S_{CDD'} - S_{D'} + S_{D'E}  \\
        S_{D} + S_{D'} & = S_{DD'}\\
        0 &= S_{CD'E} - S_C + S_{D'E} .
    \end{aligned}
\end{equation}
The first and second lines follow from the vanishing of conditional mutual information, bounded by the extreme point criterion: (1st line) $0= \Delta(C,D'E)_{\rho^h} \ge I(B:D'E|C)_{\rho^h}$, (2nd line) $0= \Delta(D',E)_{\rho^h} \ge I(CD:E|D')_{\rho^h}$. The third line is because any extreme point is factorized into a tensor product along disjoint shells arranged in the interval $\mathbb{I}$ direction. The fourth line is $\Delta(C,D'E)_{\rho^h}=0$, also by the extreme point criterion.
Therefore,
\begin{equation}
\begin{aligned}
     \Delta(B, CD' E, D)_{\rho^h} &= S_{BCD'E} + S_{CDD'E} - S_B - S_D \\
     &= (S_{BC} - S_{C} + S_{CD'E}) + (S_{D'E} - S_{D'} + S_{CDD'}) - S_B - S_D \\
     & = S_{BC} + (S_{CD'E} - S_C + S_{D'E}) + S_{CDD'}  - S_B  - (S_{D'} + S_D) \\
      &= S_{BC} + S_{CDD'} - S_B - S_{DD'}\\
      & = \Delta(B,C,DD')_{\rho^h}.
\end{aligned}
\end{equation}
The brackets in the middle steps are replacements according to Eq.~\eqref{eq:useful-for-da}.
This completes the proof.  
\end{proof}

Lemma~\ref{lemma:quantum-dim-TFRE} implies that $d_h \ge 1$, for any superselection sector $h$ detected by a sphere shell. Such superselection sectors are associated with point excitations.
It is worth noting that, in $n \ge 3$, not every type of superselection sector is associated with point excitations, and many of them are not detected by sphere shells. The general definition of quantum dimension for immersed sectorizable regions is currently lacking. 
Nonetheless, a generalization is known for regions participating in a pairing~\cite{Shi2023-Kirby}.

\section{Automorphism of information convex sets: tools and examples}\label{app:loop-space}

Automorphism (i.e., self-isomorphism) of information convex sets plays an important role in the entanglement bootstrap. The simplest application is the definition of antiparticles, but its implication is much broader.
In this appendix, we talk about two things related to automorphisms of information convex sets.

In Section~\ref{sec:complex-im}, we introduce a ``immersion complex" $\mathfrak{M}(\omega,N)$ (Definition~\ref{def:2-complex}). It is a simplicial complex whose points represent immersed regions. It is an attempt to provide a compact mathematical way to describe two things: (i) topological classes of immersion ($\omega$ into $N$) by counting the connected component ($\pi_0$) of $\mathfrak{M}$, and (ii) the topological classes of deformation of an immersed region $\Omega \looparrowright N$ back itself (keeping the intermediate configurations immersed), as the fundamental group ($\pi_1$) of a connected component of $\mathfrak{M}$. 
In Section~\ref{sec:large-deform}, we utilize tunneling (Section~\ref{sec:tunneling-trick}) to make automorphisms not available otherwise.

\subsection{Math related to immersion and deformation classes}\label{sec:complex-im}

Consider the immersion of an $n$-dimensional space with boundary into another, $\Omega \looparrowright N$. Let the homeomorphism class of $\Omega$ be $\omega$. In our convention, $\Omega$ knows all the data of the immersed region while $\omega$ forgets everything but the homeomorphism class. (We intend to first describe the mathematical problem and postpone the physical relevance to automorphisms of information convex sets and examples in the latter part of the section.)

Let $\mathcal{R}(\omega,N) = \{\Omega_1(\omega), \Omega_2(\omega), \cdots\}$ be the set of representatives of equivalence classes of immersion of $\omega$ into $N$. That is
\begin{enumerate}[itemsep=0pt]
    \item $\Omega_i(\omega) \looparrowright N$ for any $i$;
    \item $\Omega_i(\omega) \overset{N}{\not\sim} \Omega_j(\omega)$ if $i \ne j$;
    \item $\Omega_i(\omega)$ is homeomorphic to $\omega$ for any $i$;
    \item  any $\Omega \looparrowright N$ homeomorphic to $\omega$ must satisfy $\Omega \overset{N}{\sim } \Omega_i(\omega)$ for some $i$.
\end{enumerate}
Here, $Y \overset{N}{\sim } Z $ means two immersed regions $ Y, Z \looparrowright N$ can be smoothly deformed to each other on $N$, where the intermediate configurations are immersed.
In particular, if there is no way to immerse $\omega$ to $N$, $\mathcal{R}(\omega,N) = \emptyset$. Now we are ready to introduce the ``immersion complex".

\begin{definition}[Immersion complex]\label{def:2-complex}
   Let $\omega$ and $N$ be two $n$-dimensional manifolds, possibly with boundaries. We define the immersion complex $\mathfrak{M}(\omega, N)$ as 
   \begin{equation}\label{eq:M-m}
       \mathfrak{M}(\omega, N) \equiv \coprod_{\Omega \in \mathcal{R}(\omega, N)} \mathfrak{m}(\Omega, N).
    \end{equation}
    Here, $\coprod$ means disjoint union (of topological spaces), 
    and
     $\mathfrak{m}(\Omega, N)$, for any $\Omega \looparrowright N$, is the simplicial 2-complex such that 
   \begin{enumerate}
    \item  $\Omega'\looparrowright N$ represents a point (0-simplex) of $\mathfrak{m}(\Omega, {N})$  
    if and only if $\Omega' \overset{N}{\sim} \Omega$.
    \item \label{creterion-2} Two distinct points $\Omega,\Omega' \in \mathfrak{m}(\Omega, {N})$ are connected by a 1-simplex\footnote{A 1-simplex is a line segment connecting two points. One may alternatively call it a 1-cell.} of $\mathfrak{m}(\Omega, {N})$ if and only if $\Omega\cap \Omega'$ can be converted to $\Omega$ and $\Omega'$ by two sequences of extensions (one for each). 
    \item  Three distinct points $\Omega',\Omega'',\Omega''' \in \mathfrak{m}(\Omega, {N})$ are connected by a triangle (2-simplex) of $\mathfrak{m}(\Omega, {N})$ if and only if both of the following holds
    \begin{itemize}
        \item  $\Omega \cap \Omega' \cap \Omega''$ can be converted to $\Omega \cap \Omega'$, $\Omega' \cap \Omega''$ and $\Omega \cap \Omega''$ by three sequences of extensions (one for each).
        \item Points in each of the three pairs $(\Omega,\Omega')$, $(\Omega',\Omega'')$ and $(\Omega,\Omega'')$ are connected by a 1-simplex of $\mathfrak{m}(\Omega, {N})$, according to rule \ref{creterion-2} above.
    \end{itemize}      
\end{enumerate}
\end{definition}

\begin{remark}
We emphasize that the definition above serves as an attempt to formulate the mathematical problem, but we do not claim it has the desired mathematical rigor. Even if it does, we do not know if it is the most convenient to work with.
In practice, it is likely to be useful to take some cellular decomposition of manifold $N$ and only consider only ``large but finite-sized" immersed regions $\Omega$. By doing so, the number of 1-simplex connecting a given point will be finite. (This setup is also similar to entanglement bootstrap, where the manifold $N$ will be interpreted as the manifold on which we define the reference state, and thus $N$ is equipped with a coarse-grained lattice.)
\end{remark}

By definition $\mathfrak{m}(\Omega, N) = \mathfrak{m}(\Omega', N)$ if $\Omega' \overset{N}{\sim} \Omega$. Furthermore, $\mathfrak{m}(\Omega, N)$ must be connected. Therefore, $\mathfrak{M}(\omega, N)$ is the disjoint union of these connected components.
Also, by definition, for any triangle of $\mathfrak{m}(\Omega, N)$, its three edges must be included in $\mathfrak{m}(\Omega, N)$, and for any 1-simplex of $\mathfrak{m}(\Omega, N)$ its two endpoints must be included in $\mathfrak{m}(\Omega, N)$. This means both $\mathfrak{m}(\Omega, N)$ and $\mathfrak{M}(\omega, N)$ are simplicial 2-complexes.

\begin{remark}
    One may ask if it is meaningful to add higher dimensional $k$-simplexes ($k\ge 3$) instead of stopping at $k=2$. The relevance to physics we can think of so far (as we explain shortly) depends on $\pi_0$ and $\pi_1$. Adding higher-dimensional simplexes will not change them. Nonetheless, it remains an interesting question if higher homotopy (or homology) groups of $\mathfrak{M}$ can provide further information.\footnote{We thank John McGreevy for this question.}
\end{remark}

\subsubsection{Physical relevance of immersion complex through the homotopy groups}   

It is easy to explain why the number of connected components $\pi_0$ of $\mathfrak{M}(\omega, N)$ is of physical relevance. Consider the entanglement bootstrap setup where $N$ is the manifold on which we define our reference state $\sigma_N$. The $n$-dimensional version of axioms {\bf A0} and {\bf A1} are satisfied everywhere on $N$. Then $\pi_0(\mathfrak{M}(\omega, N))$, i.e., the number of connected components of $\mathfrak{M}(\omega, N)$, is the number of inequivalent immersions of the topological space $\omega$ in $N$. For instance,  $\pi_0(\mathfrak{M}(\textrm{annulus}, S^2))$ has two elements.\footnote{In general, $\pi_0(\mathfrak{M}(\omega, S^2))$ refers to the set of connected components of $\mathfrak{M}(\omega, S^2)$, and it does not automatically form a group. In the special case that $\omega$ is the \textrm{annulus}, however, there are ways to view the two components as a group $\mathbb{Z}_2$. The rest of the paper is independent of whether to make such an assignment.}  
We discussed this in the main text without using this notation. In broader contexts, $\pi_0(\mathfrak{M})$ is the set of regular homotopy classes of immersed regions, which has been a topic of study in topology; see, e.g., \cite{smale1959classification,hass1985immersions,pinkall1985regular}.  Within each connected component, $\mathfrak{m}(\Omega,N)$, $\Omega \in \mathcal{R}(\omega,N)$, we can smoothly deform the region and the isomorphism theorem implies that such deformation preserves the information convex set. Configurations in two connected components are harder to relate, although tunneling (Section~\ref{sec:tunneling-trick}) provides one way to generate isomorphisms between such configurations (when a certain Abelian sector exists).

The fundamental group of the connected components of $\mathfrak{M}$, i.e., $\pi_1(\mathfrak{m}(\Omega,N))$, by \eqref{eq:M-m} is relevant to the study of automorphism of information convex sets:  
\begin{itemize}
    \item Immersed regions up to small deformations are represented by nearby points connected by 1-simplex in $\mathfrak{m}(\Omega,N)$. For each oriented loop $\vec\gamma$ in $\mathfrak{m}(\Omega,N)$ that starts at $\Omega$ and ends at $\Omega$, we have an associated automorphism:
    \begin{equation}\label{eq:auto-gamma}
        \Phi(\vec\gamma):\quad \Sigma(\Omega) \to \Sigma(\Omega).
    \end{equation}  
    The reason \eqref{eq:auto-gamma} makes sense is that each oriented 1-simplex connecting two adjacent points $\Omega$ and $\Omega'$, defines the following isomorphism $\Sigma(\Omega) \to \Sigma(\Omega') $: we first reduce states in $\Sigma(\Omega) $ to $\Omega\cap \Omega'$ by elementary steps of restrictions, and then we extend the support of these states to $\Omega'$ by elementary steps of extensions. These steps are isomorphisms between information convex sets. After we apply these isomorphisms on $\vec\gamma$, in a sequence given by the orientation, we generate a unique automorphism \eqref{eq:auto-gamma}.

    \item Triangles (2-simplex) are the places where the loop can deform ($\vec\gamma \to \vec\gamma'$) without affecting the automorphism. We illustrate the deformation as
    \begin{equation}
        \includegraphics[width=0.58 \textwidth]{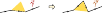}
    \end{equation} 
    where only part of the loops $\vec\gamma $ and $\vec\gamma'$ in the neighborhood of the triangle are shown.
    The fact that $\Phi(\vec\gamma) = \Phi(\vec\gamma')$ can be checked from  (Definition~\ref{def:2-complex}) and the discussion below Eq.~\eqref{eq:auto-gamma}.
\end{itemize}
We therefore argue that $\pi_1(\mathfrak{m}(\Omega,N))$ is the information that determines the nontrivial classes of deformation of immersed region $\Omega$ back to itself in the background manifold $N$. Automorphisms of $\Sigma(\Omega)$ coming from smooth deformations can be different \emph{only if} two closed loops $\vec\gamma$ and $\vec\gamma'$ (both pass $\Omega$) represent two different elements of $\pi_1(\mathfrak{m}(\Omega,N))$. In particular, if the loop $\vec\gamma$ can shrink to a point on $\mathfrak{m}(\Omega,N)$, the associated automorphism $\Phi(\vec\gamma)$ must be the (trivial) identity map. 

\subsubsection{Examples}

We give a few simple examples to illustrate the immersion complex (Definition~\ref{def:2-complex}) and discuss its physical relevance. Recall 
   \begin{equation}\label{eq:M-m-2}
       \mathfrak{M}(\omega, N) \equiv \coprod_{\Omega \in \mathcal{R}(\omega, N)} \mathfrak{m}(\Omega, N),
    \end{equation}
 where $\mathcal{R}(\omega, N)$ is the set of representatives of immersed regions homeomorphic to $\omega$, immersed in $N$. 
Our first example is about annuli immersed in a disk.
   Let $\omega$ be the annulus and $N= B^2$ be the disk.
   \begin{equation}
       \mathfrak{M}(\textrm{annulus}, B^2) = \coprod_{j=0}^{\infty} \mathfrak{m}(X(j),B^2),
   \end{equation}
   where $X(j)$ for $j \ge 0$ represents an immersed annulus with turning number $j$. In particular, we can let
    $X(0) =  X_8$ be the immersed ``figure-8" and let $X(1) = X$ be the embedded annulus. 
We expect $\mathfrak{m}(X, B^2)$ to be simply connected and $\pi_1(\mathfrak{m}(X, B^2))$ is trivial.\footnote{If we deform $X$ back to itself, keeping the intermediate configurations embedded, the deformation must correspond to a loop in $\mathfrak{m}(X, B^2)$ that can shrink to a point. The argument is precisely that used to prove Lemma~4.3 of \cite{shi2020fusion}. However, for the deformations allowing immersed intermediate configurations, the argument does not apply.}  
   What is $\pi_1(\mathfrak{m}(X_8, B^2))$? From the way to turn the figure-8 inside out by a $180^\circ$ rotation showing in Eq.~\eqref{eq:figure-8-anti-attempt}, we see that $\pi_1(\mathfrak{m}(X_8, B^2))$ has a quotient group $ \mathbb{Z}_2$. However, from Ref.~\cite{Kodama2005}, we anticipate that $\pi_1(\mathfrak{m}(X_8, B^2))= \mathbb{Z}$.\footnote{I thank John McGreevy, Joel Hass, and Ryan Budney for the reference.}

The second example is about annuli immersed in a sphere.  
   \begin{equation}
       \mathfrak{M}(\textrm{annulus}, S^2) = \mathfrak{m}(X,S^2) \coprod  \mathfrak{m}(X_8,S^2),
   \end{equation}
   where $X$ and $X_8$ are the only two classes of immersed annuli in the sphere (as discussed in Table~\ref{tab:figure-8}). Thus $ \mathfrak{M}(\textrm{annulus}, S^2) $ has two connected components. 
   What are the fundamental groups $\pi_1$ of the connected components? We anticipate that  $\pi_1(\mathfrak{m}(X,S^2))$ has a quotient group $\mathbb{Z}_2$.
   The related observation is that $X$ can turn inside out on the sphere and the fact that such a process can turn anyon $a\in \calC$ into anti-anyon $\bar{a} \in \calC$, which is, in general, a different sector. (For a related discussion, see Appendix G of \cite{shi2020fusion}.) It is worth understanding if  $\pi_1(\mathfrak{m}(X,S^2))$ itself is $\mathbb{Z}_2$.  Similiarly, for the figure-8 annulus $X_8$, it is worth understanding if $\pi_1(\mathfrak{m}(X_8,S^2))$ is $\mathbb{Z}_2$; if this is true, the definition of anti-sector on $X_8$, considered in Eq.~\eqref{eq:figure-8-anti-attempt} is justified.

   One simple example in 3-dimensional space is 
     \begin{equation}\label{eq:eversion}
       \mathfrak{M}(\textrm{sphere shell}, B^3) = \mathfrak{m}(Z,B^3),
   \end{equation}
   where $Z$, on the right-hand side, represents the embedded sphere shell in the 3-dimensional ball $B^3$. Eq.~\eqref{eq:eversion} follows from the fact that there is a unique way to immerse a sphere $S^2$ in $B^3$ (or $R^3$), up to regular homotopy; this is an influential early result by Smale \cite{smale1959classification}, which is also the topic of video~\cite{outside-in}. We are wondering if $\pi_1(\mathfrak{m}(Z,B^3)) = \mathbb{Z}_2$; if this is the case, all methods of turning a sphere inside out must be topologically equivalent. It is further interesting to consider 
        \begin{equation}\label{eq:handle-body}
       \mathfrak{M}(\textrm{genus-$g$ handlebodies}, B^3) = \mathfrak{m}(G_g,B^3),
   \end{equation}
  where $G_g$ is an embedded genus-$g$ handlebody and $g$ is a positive integer.  As is indicated in Eq.~\eqref{eq:handle-body}, $\mathfrak{M}$ is connected in this example. Genus-$g$ handlebodies detect the ``graph excitations" studied in \cite{knots-paper,Shi2023-Kirby}, where for $g=1$ these graph excitations are flux loops. In general, $\pi_1(\mathfrak{m}(G_g,B^3))$ can have nontrivial structure even though $G_g$ are sectorizable regions. 

 \begin{figure}[h]
    \centering
    \includegraphics[width=0.62\textwidth]{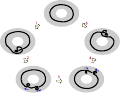}
    \caption{Illustrated is a deformation process, which corresponds to an oriented loop in $\mathfrak{m}(X, X_+)$. Here, $X$ is an annulus (thick black loop in the top figure) embedded in a wider annulus $X_+$ (gray).}
    \label{fig:annulus-in-annulus}
\end{figure}

Our last example is in 2D, where the background $N$ is an annulus. 
   Consider an annulus $X$ embedded in an annulus $N= X_+$ that thickens $X$, as shown in Fig.~\ref{fig:annulus-in-annulus}. We ask
   \begin{equation}
      \pi_1( \mathfrak{m}(X, X_+)) = ?
   \end{equation}
   Note that, $\pi_1( \mathfrak{m}(X, X_+))$ can be nontrivial even if $ \mathfrak{m}(X, B^2)$ is trivial. Intuitively, in the process described in Fig.~\ref{fig:annulus-in-annulus}, a pair of twists is created locally and then transported around the annulus before they annihilate again. (The process in Fig.~\ref{fig:annulus-in-annulus} is inspired by Fig.~8 of \cite{hass1985immersions} as well as Fig.~7 of \cite{Freedman2019}.) Can this process represent a nontrivial element of $\pi_1( \mathfrak{m}(X, X_+)) $? We do not know the answer. If this process indeed corresponds to a non-contractible loop in $\mathfrak{m}(X, X_+)$, it may have nontrivial physical consequences in contexts with topological defects.  Recall that references state on the annulus $X_+$ may come from a system with a defect line that passes through $X_+$  as discussed in Section~\ref{sec:transportation}.

In summary, the problem we try to formulate in this section is meant to be a pure topology problem. This problem is carefully separated from the study of quantum states on such regions. Nonetheless, we believe this problem is a crucial mathematical problem that will help us gain a clear understanding of the automorphism of information convex sets. 
We wonder if the existing tools in math literature are enough to produce answers to these questions. 
The question on $\pi_0(\mathfrak{M}(\omega,N))$ is essentially the classification of immersion up to regular homotopy. Some related results can be looked up\footnote{Even in this case, we do not always find the results we want. We are most interested in manifolds with boundaries (such as punctured manifolds) immersed in a manifold with \emph{the same} dimension. Sometimes, a math result we spent time searching for came out fresh~\cite{Freedman2022}.} in existing math papers e.g., \cite{smale1959classification,hass1985immersions,Kodama2005}. 
To our knowledge, the fundamental group $\pi_1(\mathfrak{m}(\Omega,N))$, e.g., for examples listed above, has not been studied in the math literature in any systematic way.  (That is why we only provided guesses to most answers!) In any case, entanglement bootstrap provides (adds) a motivation for the study of such mathematical objects.

\subsection{Automorphisms utilizing tunneling}\label{sec:large-deform}

Tunneling trick (Section~\ref{sec:tunneling-trick}) provides extra isomorphism of information convex sets by ``large changes" of the region, such as adding or removing a twist on a 1-handle. 
This is in contrast with the smooth deformation built from elementary steps. 
Tunneling is not described by lines in $\mathfrak{m}(\Omega,N)$. 
The tunneling trick, as described in the main text, works for the specific context of 2D. Below, we broaden the meaning of tunneling and use it to refer to any ``large change" of an immersed region (together with the associated information convex set).

By allowing both smooth deformations and tunneling, we can obtain more diverse classes of automorphisms of $\Sigma(\Omega)$. 
We will not discuss the entire scope of how this might work, acknowledging that it is an interesting mathematical problem. Instead, we describe a nontrivial example in 2D.

\begin{figure}[h]
    \centering
    \includegraphics[width=0.65\textwidth]{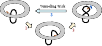}
    \caption{Automorphism of the information convex set of an immersed punctured torus, where tunneling is used. The punctured torus is visualized as an annulus (horizontal, gray) with a 1-handle (black) connecting both boundaries of the annulus.}
    \label{fig:large-deformation}
\end{figure}

Let $\mathcal{W}$ be a punctured torus immersed in a disk $B^2$, as shown in the left figure of Fig.~\ref{fig:large-deformation}. If there is an Abelian sector on figure-8 region, $\mu \in \calC_8$, then we can design an automorphism of $\Sigma(\mathcal{W})$ utilizing the tunneling trick in one of the steps; see the process shown in Fig.~\ref{fig:large-deformation}. The entire process takes the immersed punctured torus $\mathcal{W}$ back to itself. Steps 1 and 2 are smooth deformations.\footnote{As a side note, I described steps 1 and 2 in a math student seminar in 2020 \cite{shi-2020-student}; see the first backup slide, which was my attempt to ask about an unknown-to-me terminology. Later, my collaborators and I learned in the projects \cite{knots-paper,Shi2023-Kirby} that the terms I looked for were immersion and regular homotopy. Now step 3 is added, which gives a reason to ask more questions.}
Step 3 is a tunneling process. 
Thus, the automorphism cannot be associated with a closed loop in $\mathfrak{m}(\mathcal{W}, B^2)$.

Can we make use of the topological process in Fig.~\ref{fig:large-deformation} to understand nontrivial universal properties of gapped topologically ordered phases? 
We anticipate that the answer is yes. An idea is to relate this to the problem of defining and extracting topological spins from a single wave function; we postpone the discussion to future work.

\section{Strong isomorphism conjecture in arbitrary dimensions}\label{sec:SIConj-arbitrary}

Here is the statement of strong isomorphism conjecture for an arbitrary space dimension $n$, generalizing $n=2$ case (Conjecture~\ref{conj:2D-strong}) in the main text.  We denote the $n$-dimensional sphere on which we define the reference state as $\mathbf{S}^n$. Axioms {\bf A0} and {\bf A1} are satisfied everywhere.

\begin{Conjecture} 
\label{conj:strong}
    Suppose two immersed regions $\Omega,\Omega' \looparrowright \mathbf{S}^n$ are homeomorphic, then
    \begin{equation}
        \Sigma(\Omega) \cong \Sigma(\Omega').
    \end{equation}
\end{Conjecture}
Suppose this conjecture is true, then we have similar beliefs (as 2D, discussed in the discussion section) that the explicit isomorphisms that come up during the proof can be informative and useful in extracting universal information of the emergent topological quantum field theory associated with the reference state.

Besides, the strong isomorphism conjecture (Conjecture~\ref{conj:strong}) seems to resonate with another basic question: ``What is the definition of the quantum dimension for a superselection sector associated with an arbitrary immersed sectorizable region $S \looparrowright \mathbf{S}^n$?" For dimensions $n \ge 3$, this question is open.
Many (if not all) nontrivially immersed $S$ have an embedded version $S^\star \hookrightarrow \mathbf{S}^n$, where $S^\star$ is in the homeomorphism class of $S$.
Since $S^\star$ is embedded, the quantum dimension of a superselection sector  $I \in \calC_{S^\star}$ (characterized by $\Sigma(S^\star)$) can be defined from the entropy difference of an extreme point $\rho^I$ with the vacuum $\sigma$:
\begin{equation}\label{eq:d_I-familiar}
    d_I \equiv \exp{\left(\frac{S(\rho_{S^\star}^I)-S(\sigma_{S^\star})}{2}\right)}, \quad I \in \calC_{S^\star}.
\end{equation} 
While Eq.~\eqref{eq:d_I-familiar} is familiar in early works, we are still interested in knowing if $d_I \ge 1$ in general. Nontrivial checks from  3D  
and gapped domain walls  
are consistent with $d_I \ge 1$, but the general proof for an arbitrary $n$ is lacking.  
This puzzle can be restated as a question about the vacuum state: ``Does the vacuum state have the absolute minimum entropy among states in $\Sigma(\Omega)$, for any $\Omega \hookrightarrow \mathbf{S}^n$?" We would like to conjecture the affirmative direction.

  \bibliographystyle{ucsd}
\bibliography{Bibliography}  
\end{document}